\documentclass[bj]{imsart}

\RequirePackage[OT1]{fontenc}
\RequirePackage[colorlinks,citecolor=blue,urlcolor=blue]{hyperref}
\usepackage{amsmath,amstext,amssymb,amsfonts,amscd,color}
\usepackage{graphicx}
\usepackage{epsfig}
\usepackage[square, numbers]{natbib}
\usepackage{threeparttable}
\usepackage{cases,color}
\usepackage{multirow}
\usepackage{verbatim}
\usepackage{booktabs}
\usepackage{multicol}
\usepackage{multicol}
\usepackage{bbm}
\usepackage{amsthm}
\usepackage{bm}
\usepackage{subcaption}
\usepackage{lscape}
\usepackage{rotating}

\usepackage{xparse}
\usepackage{tikz}
\usetikzlibrary{plotmarks}
\usetikzlibrary{matrix}
\usepackage{pgfplots}
\usetikzlibrary{matrix,backgrounds}
\usetikzlibrary{arrows}
\usetikzlibrary{decorations.markings}
\pgfdeclarelayer{myback}
\pgfsetlayers{myback,background,main}
\tikzset{mycolor/.style = {dashed,rounded corners,line width=1bp,color=#1}}%
\tikzset{myfillcolor/.style = {draw,fill=#1}}%

\NewDocumentCommand{\highlight}{O{blue!40} m m}{%
	\draw[mycolor=#1] (#2.north west)rectangle (#3.south east);
}

\arxiv{arXiv:0000.0000}

\startlocaldefs

\def\cov{{\mbox{cov}}}
\def\var{{\mbox{var}}}

\newtheorem{assumption}{Assumption}
\newtheorem{thm}{Theorem}[section]
\newtheorem{corollary}{Corollary}[section]
\newtheorem{example}{Example}[section]
\newtheorem{remark}{Remark}[section]
\newtheorem{lemma}{Lemma}[section]

\newtheorem{proposition}{Proposition}[section]
\endlocaldefs

\begin{document}

\begin{frontmatter}
\title{Interpoint Distance Based Two Sample Tests in High Dimension}
\runtitle{Two Sample Tests in High Dimension}

\begin{aug}
\author{\fnms{Changbo} \snm{Zhu}\thanksref{a,e1}\ead[label=e1,mark]{changbo2@illinois.edu}}
\and
\author{\fnms{Xiaofeng} \snm{Shao}\thanksref{a,e2}\ead[label=e2,mark]{xshao@illinois.edu}}

\address[a]{Department of Statistics, University of Illinois at Urbana-Champaign, Champaign, IL, 61820, USA.
\printead{e1,e2}}

\runauthor{C. Zhu and X. Shao}

\affiliation{Some University and Another University}

\end{aug}

\begin{abstract}
	In this paper, we study a class of two sample test statistics based on inter-point distances in the high dimensional and low/medium sample size
setting. Our test statistics include the well-known energy distance and maximum mean discrepancy with Gaussian and Laplacian kernels, and
the critical values are obtained via permutations. We show that all these tests are inconsistent when the two high dimensional distributions
correspond to the same marginal distributions but differ in other aspects of the distributions. The tests based on
energy distance and maximum mean discrepancy mainly target the differences between marginal means and variances, whereas the
test based on $L^1$-distance can capture the difference in marginal distributions. Our theory sheds new light on the limitation of
inter-point distance based tests, the impact of different distance metrics, and the behavior of permutation tests in high dimension.
Some simulation results and a real data illustration are also presented to corroborate our theoretical
findings.
\end{abstract}


\end{frontmatter}

	\section{Introduction}
	In many statistical and machine learning applications, we need inference about the two populations or distributions based on the data samples collected. For example, we need to compare the effectiveness of two newly developed drugs in clinical research, the higher educational level between two countries in a social study and the global warming effects on two regions in environmental science. Two sample hypothesis testing is a statistical procedure to deal with such problems. Formally speaking, having i.i.d. $p$-dimensional samples $X_{1}, \cdots , X_{n} =^{d} X \sim F$ and $ Y_{1}, \cdots , Y_{m} =^d Y \sim G$, we are interested in knowing whether the underlining distributions $F$ and $G$ which generate the two samples are the same, i.e. to test the following hypothesis,
	$$
	H_{0}: F = G \text{ versus } H_{A}: F \neq G.
	$$
	
	The study of two-sample testing has a long history and dates back to Kolmogorov-Smirnov's test \cite{kolmogorov1933, smirnov1948}, where the empirical CDFs are compared using the sup-norm. Related work for univariate two-sample tests includes Cramer von-Mises criterion \cite{cramer1928,von1928} and  Anderson-Darling test \cite{anderson1952}. Extensions to comparison of multivariate distributions and also the $k$-sample problem can be found in \cite{bickel1969, bickel1983, friedman1979, henze1988,schilling1986} among others. Some other interesting work focusing on the ``trimmed'' comparison of distributions can be found in \cite{alvarez2008, alvarez2012, freitag2007, munk1998}.
	
	However, all the afore-mentioned work focuses on the fixed dimensional case. If the dimension exceeds the sample size or is allowed to grow, some of the above methods are expected to fail. For example, the density-based methods suffer from the curse of high dimensionality in particular. In this paper, we study the two sample tests based on certain dissimilarity metrics that can be expressed as functions of the interpoint distances. Two of the most popular high dimensional two-sample tests that fall into this category are based on the Energy Distance (ED) \cite{szekely2004} and the Maximum Mean Discrepancy (MMD) \cite{gretton2012}. The former is based on the Euclidean distance between sample elements; while the latter is a kernel based method and is basically a variant of ED with a user-specified kernel as distance metric. To be more specific, both ED and MMD take the following form
	\begin{equation}
	\label{Eq:}
	\text{ED}^{k}(F,G) = 2E[k (X, Y)] - E[k (X, X')] - E[k (Y, Y')],
	\end{equation}
	where $k$ is a user-specified kernel, $X', Y'$ are i.i.d copies of $X,Y$ respectively. For instance, $k$ can be chosen as
	\begin{align*}
	\begin{array}{ll}
	L^2\text{-norm} \text{ (Euclidean distance)}: & k(X,Y) = \| X-Y \|_2 = \sqrt{\sum_{u=1}^p (x_u - y_u)^2 }, \\
	\text{Gaussian kernel}: & k(X,Y) = \exp \left(- \frac{\|X-Y\|^2_2}{ 2\gamma^2_p} \right), \\
	\text{Laplacian kernel}:  & k(X,Y) = \exp \left( -\frac{\| X-Y \|_2}{ \gamma_p} \right), \\
	L^1\text{-norm} : &  k(X,Y) = \| X-Y \|_1 = \sum_{u=1}^p |x_u - y_u|,
	\end{array}
	\end{align*}
	where $X=(x_1,\cdots,x_p)^T$, $Y=(y_1,\cdots,y_p)^T$ and $\gamma_p$ is a user-specified bandwidth parameter. Then, the population version of ED is given by Equation $\eqref{Eq:}$ with $k$ being the $L^2$-norm and the population version of MMD multiplied by -1 is given by Equation $\eqref{Eq:}$ with $k$ being Gaussian or Laplacian kernel. When $k$ is $L^2$-norm, Gaussian or Laplacian kernel, ED$^k(F,G)$ enjoys the property that ED$^k(F,G)=0 \Leftrightarrow F=G$. In fact, ED$^k(F,G)=0 \Leftrightarrow F=G$ holds as long as $k$ is a strongly negative definite kernel \cite{klebanov2005n}. ED and MMD based tests are both nonparametric without any assumption on the underlying distributions and can be implemented conveniently in practice using permutations. In this work, we aim to address the following questions:
	\begin{itemize}
		\item[1,] Can $\text{ED}^k$ based permutation test maintain its power against all kinds of alternatives in the high dimensional setting?
		\item[2,] What are the impact of different distance metrics?
	\end{itemize}
	To answer the above questions, we conduct rigorous theoretical analysis on the power of $\text{ED}^k(F,G)$ based permutation test in the high dimensional low sample setting (HDLSS) \cite{hall2005} as well as high dimensional medium sample size setting (HDMSS) \cite{aoshima2018survey}. Naturally, we say a test is {\em consistent if its power goes to 1 under either HDLSS or HDMSS regime}. Here, we study the power property of the permutation based tests because they are frequently implemented for Energy Distance and its variants in real life applications. 
	
	Let $\mathbf{X} = ( X_{1},X_{2}, \cdots, X_{n} )^{T}$, $\mathbf{Y} = ( Y_{1},Y_{2}, \cdots, Y_{m} )^{T}$, $\mathbf{Z} = (\mathbf{X}^T, \mathbf{Y}^T)^T$ denote the sample matrices and $\text{ED}_n^k(\mathbf Z)$ be a U-statistic based unbiased estimator of $\text{ED}^k(F,G)$. Our main results include: (i) Derivation of the limiting distribution of $\text{ED}_n^k(\mathbf{\Gamma} \mathbf Z)$ under both low and medium sample size setting, where  $\mathbf \Gamma \sim \text{Uniform}(\mathbb P_{n+m})$ and $\mathbb{P}_{n+m}$ is the set of permutation matrices of dimension $(n+m) \times (n+m)$. (ii) Based on the asymptotic results, we formulate different local alternatives, under which the power behavior of $\text{ED}_n^k$ based permutation tests are discussed in detail. (iii) Our theories are applied to existing kernels and statistics, for example 
	\begin{itemize}
		\item[1,] Under both HDLSS and HDMSS, ED$^k$ based permutation test w.r.t. $L^2$-norm, Gaussian and Laplacian kernel are consistent if the sum of component-wise mean or variance differences are not so small, i.e, $\lim_{p \rightarrow \infty}\sum_{u=1}^{p} (E(x_u)-E(y_u))^2/p \neq 0$ or $\lim_{p \rightarrow \infty} |\sum_{u=1}^{p} (var(x_u) - var(y_u))/p| \neq 0$. In addition, if the sum of component-wise mean and variance differences are both of order $o(\sqrt{p}/\sqrt{nm})$, i.e., $$\sum_{u=1}^{p} (E(x_u)-E(y_u))^2 = o\left(\frac{\sqrt{p}}{\sqrt{nm}}\right) \text{ and } \left|\sum_{u=1}^{p} (var(x_u) - var(y_u)) \right| = o\left(\frac{\sqrt{p}}{\sqrt{nm}}\right),$$ these tests suffer substantial power loss (the limits of their power are derived) under HDLSS and have trivial power (power no larger than the significance level) under HDMSS. Furthermore, under HDLSS, the afore-mentioned tests have trivial power if additionally we have $\sum_{u,v=1}^{p} (\cov(x_u, x_v) - \cov (y_u, y_v) )^2 = o(p).$
		\item[2,]  When $k$ is chosen as $L^1$-norm, ED$^k$ based permutation test experiences a power drop under HDLSS  and trivial power under HDMSS if $X,Y$ have the same univariate marginal distribution, i.e. $x_u =^d y_u$ for $u = 1, 2, \cdots, p$. This phenomenon is consistent with the fact that ED$^k$ with $L^1$-norm can characterise the discrepancies between the marginal univariate distributions. In addition, Under HDLSS, we show that the $L^1$-norm based test has trivial power when $X$ and $Y$ have the same bivariate marginal distribution, i.e., $(x_u, x_v) =^d (y_u, y_v), u, v  =1, \cdots, p. $
	\end{itemize}
	These findings are further corroborated in our simulation study. It is worth mentioning that \citet{chakraborty2019new} investigate the energy distance, maximum mean discrepancy, distance covariance and Hilbert-Schmidt Independence Criterion in the high dimensional setting. They propose a new class of metrics which can detect/measure the equality of   low-dimensional marginal distributions and a computational efficient $t$-test is further proposed based on the new metric. By contrast, our focus is on kernel-based permutation test and their asymptotic power properties in the high dimensional setting.  In the following we introduce some  notation and define some frequently used operators for later convenience.
	
	\subsection{Notation}
	Here, random data samples are denoted as, for each $i = 1, 2, \cdots, n$, $X_{i} =^d X= (x_{1}, \cdots, x_{p})^T \in \mathbb{R}^{p}$ and for each $j=1,2, \cdots, m$, $Y_{j} =^d Y= (y_{1}, \cdots, y_{p})^T \in \mathbb{R}^{p}$. Next, let $\mathbf{X} = ( X_{1},X_{2}, \cdots, X_{n} )^{T}$, $\mathbf{Y} = ( Y_{1},Y_{2}, \cdots, Y_{m} )^{T}$ and $\mathbf{Z} = (\mathbf{X}^T, \mathbf{Y}^T)^T = (Z_{1}, Z_2,  \cdots, Z_{n+m})^T$ denote the random sample matrices.
	Furthermore, let $\mathbb{P}_{n+m} = \left\{ \Gamma_{1}, \Gamma_{2}, \cdots, \Gamma_{(n+m)!}  \right\}$ be the group containing all  permutation matrices of dimension $(n+m) \times (n+m)$ and for each $i $, let $\pi_i$ be the permutation that corresponds to $ \Gamma_i $ via
	\begin{align*}
	\Gamma_i \mathbf{ Z} =  \left(Z_{(\pi_i(1))}, Z_{(\pi_i(2))},  \cdots, Z_{(\pi_i(n+m))}\right)^T
	\end{align*}
	where $ (\pi_i(1)) < \cdots < {(\pi_i(n+m))} $ is the ranked sequence of $ \{ \pi_i(1), \pi_i(2), \cdots, \pi_i(n+m) \} $. For a random permutation matrix $\bm{\Gamma} \sim \text{uniform}(\mathbb{P}_{n+m})$, we use $\bm{\pi}$ to represent its corresponding permutation. Next, given any function $\varphi$, $\varphi^{(i)}$ is used to denote its $i$-th order derivative. Finally, calligraphic letters ($\mathcal{K}, \mathcal{L}, \mathcal{R}, \mathcal{W}, \mathcal{G}$) are used to denote self-defined operators that act on random variables to produce random variables.  
	
	\section{Interpoint Distance Based Two Sample Tests}
	\label{sectionED}
	In this paper, we limit our attention to $\text{ED}^k(F,G)$, where $k$ is a user specified dissimilarity metric \cite{sarkar2018high} of the following form
	\begin{equation}
	\label{Eq:metric}
	k (X,Y) =\varphi \left\lbrace \frac{1}{p}  \sum\limits_{u=1}^{p} \psi (x_{u}, y_{u}) \right\rbrace,
	\end{equation}
	where $\psi \geq 0$ and $\varphi$ has continuous second order derivative on $(0, +\infty)$. The reason we focus on $\text{ED}^k(F,G)$ of the above form is that the metric $k$ encompasses many well-known distance metrics such as $L^2$-norm, $L^1$-norm, Gaussian and Laplacian kernel. Consequently, Energy Distance (ED) and Maximum Mean Discrepancy (MMD) are just special cases of $\text{ED}^k(F,G)$. We summarize the commonly used distance metrics in Table \ref{Tab:cor}. Following the literatures \cite{gretton2012,gretton2007}, we consider the bandwidth parameter $\gamma$ in Gaussian and Laplacian kernel as a fixed constant.
	\begin{table}
		$$
		\def\arraystretch{1.5}
		\begin{array}{|c|c|c|c|} \hline
		\psi(x,y) & \varphi (x)	  &\def\arraystretch{1.2}\begin{array}{c}
		k 
		\end{array}  & \def\arraystretch{1.2}\begin{array}{c}
		\text{ED}^{k}(F,G) 
		\end{array}
		\\ \hline \hline
		\multirow{4}{*}{$ (x-y)^{2}$} & \sqrt{x}  & L^2\text{-norm} & \begin{array}{c}
		\text{Energy distance (ED)} \\
		\text{\citet{szekely2004}}
		\end{array}   \\ \cline{2-4}
		& - e^{- \frac{x}{ 2 \gamma^2} }  &
		\def\arraystretch{1.2}\begin{array}{c}
		\text{Gaussian kernel} \\ \text{(multiplied by -1)} \end{array} & \multirow{2}{*}{$ \begin{array}{c}
			\text{Maximum Mean Discrepancy (MMD)} \\
			\text{\citet{gretton2012}}
			\end{array} $} \\ \cline{2-3}
		& - e^{- \frac{\sqrt{x}}{ \gamma }} & \def\arraystretch{1.2}\begin{array}{c}
		\text{Laplacian kernel} \\ \text{(multiplied by -1)}
		\end{array} & \\ \hline
		|x-y| & x & L^1\text{-norm}  & \begin{array}{c}
		\text{Used for some graph-based tests} \\
		\text{\citet{sarkar2018high}}
		\end{array}   \\  \hline
		\end{array}
		$$
		\caption{Correspondence between different choices of $\psi, \varphi$ and existing distance metrics as well as two sample test statistics in the literature.}
		\label{Tab:cor}
	\end{table}
	Notice that if $k$ is some well-known distance metrics such as $L^2$-norm, Gaussian kernel (multiplied by -1) and Laplacian kernel (multiplied by -1), a nice property for $\text{ED}^k$ is that
	\begin{align}
	\label{eq:nice}
	\text{ED}^{k} (F, G) \geq 0 \text{ and }\text{ED}^{k} (F, G) =  0 \Leftrightarrow F= G.
	\end{align}
	Here, it is just for the ease of presentation and notational simplicity that $k$ is set to be Gaussian or Laplacian kernel multiplied by -1. In fact, if $k$ is a universal kernel (see Theorem 5 and Lemma 1 of \cite{gretton2012}) or $k$ is a strongly negative definite kernel (see Theorem 1.9 \cite{klebanov2005n}), Property \eqref{eq:nice} still holds. On the other hand, using $\text{ED}^1(F,G)$ to denote $\text{ED}^k(F,G)$ when $k$ is the $L^1$-distance, we observe that $\text{ED}^{1}(F,G) = \sum_{u=1}^p \text{ED}(F_u, G_u)$, from which it  easily follows that
	\begin{align*}
	\text{ED}^{1}(F,G) \geq 0 \text{ and } \text{ED}^{1} (F, G) =  0 \Leftrightarrow F_u = G_u \text{ for all }u=1,2, \cdots, p.
	\end{align*}
	Notice that it is possible to have $F_u = G_u \text{ for all }u=1,2, \cdots, p$ but $F \neq G$, under which we have $\text{ED}^{1}(F,G) = 0$ while $\text{ED}^k >0$ if $k$ is $L^2$-norm, Gaussian kernel (multiplied by -1) or Laplacian kernel (multiplied by -1). Thus, $L^2$-norm, Gaussian kernel or Laplacian kernel based test statistics have advantage over $L^1$-norm based test statistic in the low dimensional setting, but we will see later that the story is in a sense reversed under the high dimensional setting. Next, an unbiased estimator of $\text{ED}^k$ is given as
	\begin{multline*}
	\text{ED}_n^{k} (\mathbf{Z}) =  \frac{2}{mn} \sum\limits_{i=1}^{n} \sum\limits_{j=1}^{m} k(X_i, Y_j) \\ - \frac{2}{n(n-1)}  \sum\limits_{1 \leq i < j \leq n} k(X_i, X_j) -  \frac{2}{m(m-1)} \sum\limits_{1 \leq i < j \leq m} k(Y_i, Y_j).
	\end{multline*}
	
	\section{Power Analysis for Permutation Test}
	As permutation tests are commonly used for Energy Distance and kernel variants in practice due to their implementational convenience and accurate size, we study their asymptotic behavior under the high dimensional setting in this subsection. Since we have i.i.d samples, after we permute the data, i.e., shuffle the rows of $\mathbf Z$ as $ \Gamma_{i} \mathbf Z $ by some permutation matrix $\Gamma_i$, what really matters to the distribution of $\text{ED}^k_n(\Gamma_i \mathbf Z)$ is how many $X$ samples stay in the first $n$ rows. Formally, let $ |\mathbb{A}|  $ be the cardinality of the set $\mathbb{A}$ and given a permutation matrix $\Gamma_i$ with the corresponding permutation $\pi_i$, set
	$$
	N(\Gamma_i) = \left|  \left\lbrace j \in \{1,2, \cdots, n \} : 1 \leq j \leq n, n+1 \leq \pi_i (j) \leq n+m \right\rbrace \right|.
	$$
	The integer $ n- N(\Gamma_i) $ actually counts the number of samples which belong to the first $n$ rows of $\mathbf{Z}$ both before and after the permutation $\Gamma_{i}$. Notice that it is possible that $N(\Gamma_i) = N(\Gamma_j)$ for different permutations $ \Gamma_i $ and $ \Gamma_j $. The set $ \mathbb{S}_w $ collects all the permutations $ \Gamma_i $ such that $N(\Gamma_i) = w$. Mathematically, fix $0 \leq w \leq \text{min} \{ n,m \} $, set $\mathbb{S}_{w} =\left\{  \Gamma_{i} : N(\Gamma_{i}) = w, \; i = 1,2, \cdots, (n+m)! \right\},$ then
	$$
	| \mathbb{S}_{w} | = \binom{m}{w}\binom{n}{n-w}n!m!.
	$$
	To differentiate from $\Gamma_i$, we use italic symbol $\varGamma_w$ to represent an element in $\mathbb{S}_w$.  Intuitively, $| \mathbb{S}_w |$ is the number of permutations that would have $n-w$ samples stay in the first $n$ rows of $\mathbf Z$ after we apply the corresponding permutation. The above process is further illustrated in the following diagram.

	\pgfdeclarelayer{background}
	\pgfdeclarelayer{foreground}
	\pgfsetlayers{background,main,foreground}

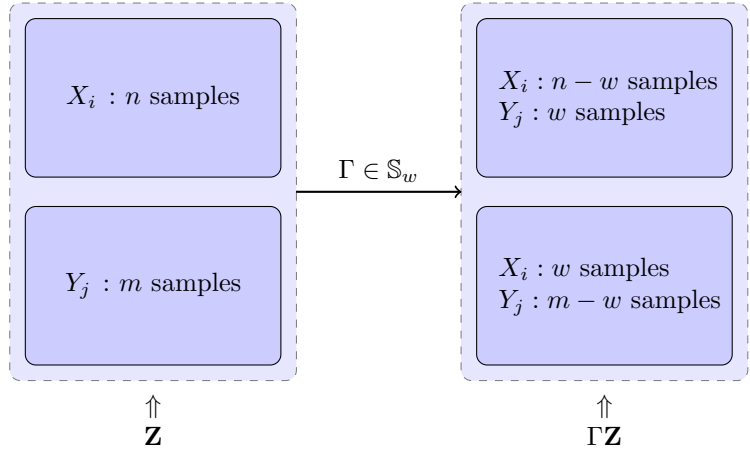
	\tikzstyle{sensor}=[draw, fill=blue!20, text width=9em,
	text centered, minimum height=6em, rounded corners]
	\tikzstyle{ann} = [above, text width=5em]
	\tikzstyle{naveqs} = [sensor, text width=6em, fill=red!20,
	minimum height=12em, rounded corners]
	\def\blockdist{2.3}
	\def\edgedist{2.5}
	\begin{center}
		\begin{tikzpicture}[fill=blue!20]
		\path (-3,1.5) node(a) [sensor] {$X_i:$ $n$ samples}
		(-3,-1) node(b) [sensor] {$Y_j:$ $m$ samples}
		(3,1.5) node(c) [sensor] {$\begin{array}{l}
			X_i: n -w \text{ samples} \\
			Y_j : w \text{ samples }
			\end{array}$}
		(3,-1) node(d) [sensor] {$\begin{array}{l}
			X_i: w \text{ samples} \\
			Y_j : m-w \text{ samples }
			\end{array}$};
		\node [below of = b, node distance= 1.6cm] (l) {$ \Uparrow $ };
		\node [below of = l, node distance= 0.4cm] (u) {$\mathbf Z$ };
		\node [below of = d, node distance= 1.6cm] (m) {$ \Uparrow $ };
		\node [below of = m, node distance= 0.4cm] (v) {$\Gamma \mathbf Z$ };
		
		\path [draw, ->, thick] (-1.1, 0.25) --node[above]{$\Gamma \in \mathbb{S}_w$}  (1.1, 0.25) ;
		
		\begin{pgfonlayer}{background}
		\path (a.west |- a.north)+(-0.2,0.2) node (e) {};
		\path (b.south -| b.east)+(+0.2,-0.2) node (f) {};
		\path[fill=blue!10,rounded corners, draw=black!50, dashed]
		(e) rectangle (f);
		\path (c.west |- c.north)+(-0.2,0.2) node (g) {};
		\path (d.south -| d.east)+(+0.2,-0.2) node (h) {};
		\path[fill=blue!10,rounded corners, draw=black!50, dashed]
		(g) rectangle (h);
		\end{pgfonlayer}
		\end{tikzpicture}
	\end{center}
	For the inter-point distance based two sample tests,  we can equivalently permute the weights on the pair-wise distances instead of permuting data points, i.e.,  for a fixed permutation matrix $ \Gamma_s \in \mathbb{P}_{n+m} $ that corresponds to $\pi_s$, we can write $ \text{ED}_n^k( \Gamma_s \mathbf{Z}) $ as 
	\begin{align}
	\label{eq:permuteWeight}
	\text{ED}_n^{k}(\Gamma_s \mathbf{Z}) = \sum_{i=2}^{n+m}\sum_{j=1}^{i-1} \Pi_{s,ij} k(Z_i, Z_j),
	\end{align}
	where $\Pi_{s,ij}$ is defined as
	\begin{align*}
	\Pi_{s,ij} = \left\lbrace  \begin{array}{ll}
	-\frac{2}{n(n-1)}, & 1 \leq \pi_s(i), \pi_s(j) \leq n, \\
	-\frac{2}{m(m-1)}, & n+ 1 \leq \pi_s (i), \pi_s(j) \leq n +m,  \\
	\frac{2}{mn}, &  1 \leq \pi_s(i) \leq n, n+1 \leq  \pi_s(j) \leq n +m, \\
	\frac{2}{mn}, &  n+1 \leq \pi_s(i) \leq n+m, 1 \leq  \pi_s(j) \leq m.
	\end{array}  \right. 
	\end{align*}
	To formally define the permutation test for $ \text{ED}^k_n(\mathbf Z) $, let $\widehat{R}$ denote the randomization distribution of  $ \text{ED}_n^k(\mathbf{Z}) $, which is defined by
	\begin{align*}
	\widehat{R}(t)  = \frac{1}{(n+m)!} \sum\limits_{i=1}^{(n+m)!} \mathbb{I}_{ \left\lbrace \text{ED}_n^k(\Gamma_{i} \mathbf{Z})  \leq t \right\rbrace }  = \frac{1}{(m+n)!} \sum\limits_{w=0}^{\min \{n,m \}} \sum\limits_{ \Gamma \in \mathbb{S}_{w}}  \mathbb{I}_{\left\lbrace   \text{ED}_n^k(\Gamma \mathbf{Z})  \leq t \right\rbrace}.
	\end{align*}
	For any distribution $F$, let the $(1-\alpha)$-th quantile of $F$ be denoted by $Q_{F,1-\alpha}$. In particular, the $(1-\alpha)$th quantile of $\widehat{R}$ is $Q_{\widehat{R},1-\alpha}$, i.e.
	\begin{equation}
	\label{quantile}
	Q_{\widehat{R},1-\alpha} = \widehat{R}^{-1} (1 - \alpha) = \inf \left\lbrace t: \widehat{R}(t) \geq 1 - \alpha \right\rbrace.
	\end{equation}
	Then, the level-$\alpha$ permutation test w.r.t. $\text{ED}_n^k(\mathbf Z)$ is defined as
	$$
	\text{Reject } H_0 \text{ if } \text{ED}_n^k(\mathbf{Z}) > Q_{\widehat{R},1-\alpha}. 	
	$$
	In real life applications, $ (n+m)! $ might be large, we thus resort to an approximation of $ Q_{\widehat{R} , 1-\alpha} $. Let  $ \mathbf{\Gamma}_1, \cdots, \mathbf{\Gamma}_S $ be i.i.d and uniformly sampled from $ \mathbb{P}_{n+m} $ and we approximate the critical value by $Q_{\widetilde{R}, 1-\alpha}$, where
	\begin{align*}
	\widetilde{R}(t):=\frac{1}{S} \left(\mathbb{I}_{ \left\lbrace \text{ED}_n^k( \mathbf{Z})  \leq t \right\rbrace } + \sum\limits_{i=1}^{S-1}  \mathbb{I}_{ \left\lbrace \text{ED}_n^k(\mathbf \Gamma_{i} \mathbf{Z})  \leq t \right\rbrace }  \right).
	\end{align*}
	
	\subsection{Local Alternatives}
	\label{subsec:local}
	In this subsection, we define different local alternatives, under which the $\text{ED}_n^k(\mathbf{Z})$ based permutation test will be consistent, have a nontrivial power limit and exhibit trivial power (power  no larger than the significance level $\alpha$) in the limit. To formally define the local alternative hypothesis, let the operator $\mathcal{K}$ be defined as 
	\begin{multline}\label{eq:K}
	\mathcal{K}(Z_i, Z_j)  = \frac{1}{\sqrt{p}}  \sum_{u=1}^p \big\{ \psi (z_{iu}, z_{ju}) - E\left[\psi (z_{iu}, z_{ju})|z_{iu}\right] \\ -  E\left[\psi (z_{iu}, z_{ju})|z_{ju}\right] +  E\left[\psi (z_{iu} , z_{ju})\right] \big\},
	\end{multline}
	It follows from Proposition 2.2.1 of \cite{zhu2019distance} that
	$
	E[\mathcal{K}(Z_{i},Z_{j}) \mathcal{K}(Z_{i'},Z_{j'}) ] = 0$  if $  \{i, j \} \neq \{ i', j' \}.   
	$
	Next, denote the average distance over components as 
	$$
	\overline{\psi}(Z_i,Z_j) = \frac{1}{p} \sum_{u=1}^{p} \psi (z_{iu}, z_{ju}).
	$$ 
	In addition, we need to assume the existence of some constants to properly define the local alternatives. These constants will also appear in the limiting distribution of our test statistics.
	\begin{assumption} 
		\label{ass:1}
		As $p \rightarrow \infty$, assume the existence of the limiting mean
		\begin{align*}
		e_{x}  =  \lim\limits_{p \rightarrow \infty}  E \left[ \overline{ \psi} (X, X') \right]  ,
		e_{y}  =  \lim_{p \rightarrow \infty}  E \left[ \overline{\psi} (Y, Y') \right] 
		\text{ and } e_{xy}  =  \lim\limits_{p \rightarrow \infty} E \left[\overline{ \psi} (X, Y) \right]
		\end{align*}
		and also the limiting variances
		\begin{align*}
		v_x = \lim_{p \rightarrow \infty} var \left[ \mathcal{K}(X, X') \right], v_y = \lim_{p \rightarrow \infty} var \left[ \mathcal{K}(Y, Y') \right] \text{ and } v_{xy} = \lim_{p \rightarrow \infty} var \left[ \mathcal{K}(X, Y) \right].
		\end{align*}
	\end{assumption}
	Then, we are ready to define the consistency space $H_{A_c}$, under which the $\text{ED}_n^k(\mathbf{Z})$ implemented as permutation test can be shown to be consistent under both HDLSS and HDMSS settings.
	\begin{align*}
	H_{A_c} := \left\lbrace (F, G) \; | \; 2 \varphi(e_{xy}) \neq  \varphi(e_{x})+  \varphi(e_{y})  \right\rbrace .
	\end{align*}
	We use $\mathbb{A}^c$ to denote the complement of any given set $\mathbb{A}$ and denote $F= (F_1, F_2, \cdots, F_p)$ and $  G= (G_1, G_2, \cdots, G_p) $, where $ F_u, G_u, u =1,2, \cdots, p $ are the marginal univariate distributions.  For commonly used kernels,  we have the following table characterizing $H_{A_c}$ and the proof is postponed to subsection \ref{subsec:hac}. 
	$$
	\def\arraystretch{1.5}
	\begin{array}{|c|c|} \hline
	k  & H_{A_c} \text{ Characterization}   \\ \hline \hline
	L^2\text{-norm} & \multirow{3}{*}{ $H_{A_c} = \left\lbrace (F,G) \left|
		\begin{array}{c}
		\sum_{u=1}^{p} (E(x_u)-E(y_u))^2 = o(p) \text{ and} \\
		\left|\sum_{u=1}^{p} (var(x_u) - var(y_u))\right| = o(p)
		\end{array} \right.
		\right\rbrace^c
		$}   \\ \cline{1-1}
	\text{Gaussian kernel} &   \\ \cline{1-1}
	\text{Laplacian kernel} &   \\ \hline
	L^1\text{-norm}
	&
	H_{A_c} =  \left\lbrace (F,G) \left| \sum_{u=1}^p \text{ED}(F_u, G_u) = o(p) \right. \right\rbrace^c 
	\\  \hline
	\end{array}
	$$
	Then, we present the space $H_{A_l}$, under which the normal limit of $\text{ED}_n^k(\mathbf{Z})$ can be derived under both HDLSS and HDMSS. 
	\begin{align*}
	H_{A_{l}} := \left\lbrace  (F,G) 
	\left| 
	\begin{array}{l}
	e_{xy} = e_x = e_y,  \\ 
	\left| 2 E\left[\overline{\psi}(X,Y)\right] -  E\left[\overline{\psi}(X,X')\right] - E\left[\overline{\psi}(Y,Y')\right] \right| = o( \sqrt{\frac{1}{nmp}}  ) ,  \\
	E \left[ \left|  E\left[\overline{\psi}(X,Y)|X \right] -E\left[\overline{\psi}(X,X')|X \right] \right| \right] = o(\sqrt{\frac{1}{nmp}}  ) \text{ and} \\
	E \left[ \left|  E\left[\overline{\psi}(X,Y)|Y \right]  -E\left[\overline{\psi}(Y,Y')|Y \right] \right| \right]  = o( \sqrt{\frac{1}{nmp}} ).
	\end{array}
	\right.
	\right\rbrace .
	\end{align*}
	Under $H_{A_l}$, a limit for the power of $\text{ED}_n^k(\mathbf{Z})$ (implemented as permutation test) is derived under HDLSS. On the other hand, its power is shown to be trivial (no larger than the significance level $\alpha$) under HDMSS and $H_{A_l}$. Next, we provide sufficient conditions for $(F,G) \in H_{A_l}$ with respect to the following well known kernels. 
	$$
	\def\arraystretch{1.5}
	\begin{array}{|c|c|} \hline
	k  &   \text{Sufficient conditions for }H_{A_l}   \\ \hline \hline
	L^2\text{-norm} & \multirow{3}{*}{ $ \left\lbrace (F,G) \left|
		\begin{array}{l}
		\sum_{u=1}^{p} (E(x_u)-E(y_u))^2 = o(\sqrt{\frac{p}{nm}}) \text{ and} \\
		\left|\sum_{u=1}^{p} (var(x_u) - var(y_u))\right| = o(\sqrt{\frac{p}{nm}})
		\end{array} \right.
		\right\rbrace \subseteq H_{A_l}
		$}   \\ \cline{1-1}
	\text{Gaussian kernel} &   \\ \cline{1-1}
	\text{Laplacian kernel} &   \\ \hline
	L^1\text{-norm}
	&
	\left\lbrace (F,G) \left| F_u = G_u, u = 1,2, \cdots, p \right. \right\rbrace \subseteq H_{A_l}
	\\  \hline
	\end{array}
	$$
	Then, the set of distributions $H_{A_t}$ is defined as 
	\begin{align*}
	H_{A_t} := \left\lbrace   (F,G) \left| (F,G) \in H_{A_l}, v_{xy} = v_{x} = v_{y} \right. \right\rbrace.
	\end{align*}
	It can be shown that under $H_{A_t}$, the $\text{ED}_n^k(\mathbf{Z})$ based permutation test has power no larger than the significance level $\alpha$ for both HDLSS and HDMSS settings. Sufficient conditions of being in $ H_{A_t} $ are provided in the following table.
	$$
	\def\arraystretch{1.5}
	\begin{array}{|c|c|} \hline
	k  &  \text{Sufficient conditions for } H_{A_t}  \\ \hline \hline
	L^2\text{-norm} & \multirow{3}{*}{ $ \left\lbrace (F,G) \left|
		\begin{array}{l}
		\sum_{u=1}^{p} (E(x_u)-E(y_u))^2 = o(\sqrt{\frac{p}{nm}}), \\
		\left|\sum_{u=1}^{p} (var(x_u) - var(y_u))\right| = o(\sqrt{\frac{p}{nm}}) \text{ and} \\
		\sum_{u,v=1}^{p} (\cov(x_u, x_v) - \cov (y_u, y_v) )^2 = o(p)
		\end{array} \right.
		\right\rbrace \subseteq H_{A_t}
		$}   \\ \cline{1-1}
	\text{Gaussian kernel} &   \\ \cline{1-1}
	\text{Laplacian kernel} &   \\ \hline
	L^1\text{-norm}
	&
	\left\lbrace (F,G) \left| (x_u, x_v) =^d (y_u, y_v), u, v  =1, \cdots, p  \right. \right\rbrace \subseteq H_{A_t}
	\\  \hline
	\end{array}
	$$
	Comparing the three local alternatives, it follows from the definition of $H_{A_c}, H_{A_l}, H_{A_t}$ that $H_{A_c}^c \supseteq H_{A_l} \supseteq H_{A_t}$.  We also want to remark that it holds for arbitrary function $\varphi$ and $\psi$ that 
	\begin{align*}
	\left\lbrace (F,G) \left| F_u = G_u, u = 1,2, \cdots, p \right. \right\rbrace \subseteq H_{A_l}, \\
	\left\lbrace (F,G) \left| (x_u, x_v) =^d (y_u, y_v), u, v  =1, \cdots, p  \right. \right\rbrace \subseteq H_{A_t}.
	\end{align*}

	\subsection{High Dimensional Low Sample Size (HDLSS)}
	The analysis in this subsection is conducted under the high dimensional low sample size setting (HDLSS), i.e, $n,m$ are fixed constants and we let $p \rightarrow \infty$. Our final goal is to study the power of $\text{ED}_{n}^{k}( \mathbf Z)$ based permutation test under various local alternatives. To this end, we need the following assumption. Recall the operator $\mathcal{K}$ is defined in \eqref{eq:K}.
\begin{assumption} \label{Ass:1}
	For fixed $n,m$, as $p \rightarrow \infty$,
	\begin{align*}
	\left(
	\begin{array}{c}
	\mathcal{K}(X_i, Y_j)   \\
	\mathcal{K}(X_{i_1}, X_{i_2})  \\
	\mathcal{K}(Y_{j_1}, Y_{j_2})  \\
	\end{array} \right)_{i,j,  i_1 < i_2  , j_1 < j_2 }  \overset{d}{\rightarrow} \left(
	\begin{array}{c}
	b_{ij} \\
	c_{i_1 i_2} \\
	d_{j_1 j_2}
	\end{array} \right)_{ i,j,  i_1 < i_2  , j_1 < j_2 },
	\end{align*}
	where $\{b_{ij}, c_{i_1 i_2}, d_{j_1 j_2} \}_{ i,j,  i_1 < i_2  , j_1 < j_2}$ are uncorrelated and jointly Gaussian with mean 0 and variances $\var( b_{ij}) = v_{xy} $, $\var(c_{i_1 i_2}) = v_{x}$, $\var(d_{j_1 j_2}) = v_{y}$.
\end{assumption}
	\begin{remark}
		The above multi-dimensional CLT result is classical and can be derived under suitable moment and weak dependence assumptions on the components of $X$ and $Y$. 
	\end{remark}
	In the above assumption, it is due to the use of double centered distance $\mathcal{K}(Z_i, Z_j)$ that the asymptotic covariance matrix is diagonal. Then, to provide some insights, the first step of our power analysis is the Taylor expansion w.r.t $\varphi$ up to the second order, i.e., for $i \neq j$
	\begin{align*}
	k(Z_i, Z_j)  = \varphi\left( e_{ij} \right) + \varphi^{(1)}(e_{ij}) \mathcal{L}(Z_i, Z_j)  + \mathcal{R}_{2}(Z_i,Z_j),
	\end{align*}
	where $
	\mathcal{L}(Z_i, Z_j) :=  \overline{\psi} (Z_{i}, Z_{j}) - e_{ij}
	$ is an operator that acts on random variables, $\mathcal{R}_{2}(Z_i,Z_j) $ is the remainder and 
	\begin{align*}
	e_{ij} = \left\lbrace  \begin{array}{ll}
	e_{x}, & \text{ if } 1 \leq i, j \leq n, \\
	e_{y}, & \text{ if } n+1 \leq i, j \leq n+m, \\
	e_{xy}, & \text{ otherwise}.
	\end{array} \right. 
	\end{align*}
	In order to control the remainder term, we need assumptions about the decay rate of $E[ \mathcal{L}^2(Z_i, Z_j) ]$. Thus, we set 
	$$
	\alpha_{x}^2 = E[\mathcal{L}^2(X,X')], ~\alpha_{y}^2 = E[\mathcal{L}^2(Y,Y')] \text{ and }  \alpha_{xy}^2 = E[\mathcal{L}^2(X,Y)].
	$$ 
	It then follows from Markov's inequality that $ \mathcal{L}(X,Y) = O_p(\alpha_{xy}) $, $ \mathcal{L}(X,X') = O_p(\alpha_{x}) $ and $ \mathcal{L}(Y,Y') = O_p(\alpha_{y}) $. Then, our next two assumptions are used to control the remainder terms induced by taking the Taylor expansion.
	\begin{assumption} \label{Ass:3}
		$ 
		\alpha_{xy}^2 = o(1),  \alpha_{x}^2 = o(1) \text{ and } \alpha_{y}^2 = o(1).
		$
	\end{assumption}
	\begin{assumption} \label{Ass:4}
		$
		\sqrt{p} \alpha_{xy}^2 = o(1), \sqrt{p} \alpha_{x}^2 = o(1) \text{ and } \sqrt{p}\alpha_{y}^2 = o(1).
		$
	\end{assumption}
\begin{remark}\label{rmk:l1}
	To gain some insights into the above assumptions, a straightforward calculation yields
	\begin{align*}
	\alpha_{xy}^2 & = \frac{1}{p^2} \sum_{u,v=1}^p \cov( \psi(x_u, y_u), \psi(x_v, y_v) ) + \left(\frac{\sum_{u=1}^{p} E[ \psi (x_{u}, y_{u})]}{p}- e_{xy}\right)^2 .
	\end{align*} 
	Therefore, we have $ \sqrt{p} \alpha_{xy}^2 = o(1) $ if the component-wise dependencies of both $X$ and $Y$ are not so strong. For illustration purpose, suppose $X$ and $Y$ are $\kappa$-dependent weak stationary time series, i.e., $x_u \perp x_{v}$ and $y_u \perp y_v$ if $|u-v| > \kappa$. Then, if $\max_u E [ \psi^2(x_u, y_u)]< \infty$, it is easy to see that $ \alpha_{xy}^2 = O\left(\kappa/p \right)$ and as a consequence, Assumption \ref{Ass:4} is satisfied as long as $\kappa / \sqrt{p} = o(1)$. In addition, it is indeed fairly straightforward to verify the above result when the sequence $\{(x_u, y_u)\}_{u=1}^p$ is $\alpha$-mixing with geometrically decaying coefficients.
\end{remark}
\begin{remark}
	When $\psi(x,y) = (x-y)^2$, some algebra shows that
	\begin{multline*}
	 \sum_{u,v=1}^p \cov( \psi(x_u, y_u), \psi(x_v, y_v) )  
	=  \var(X^T X) + \var(Y^T Y) \\+ 4 \var(X^TY) 
	  - 4 \cov(X^TX, X^T E[Y]) - 4 \cov(Y^TY, Y^T E[X]).
	\end{multline*}
	Thus, suppose $ \sum_{u=1}^{p} E[ \psi (x_{u}, y_{u})]/p-e_{xy} = o(p^{-1/4}) $ and if $\var(X^T X)$, $\var(Y^T Y),$  $\var(X^TY),$  $\var(X^TE[Y])$, $ \var(E[X]^TY)$ all have order $o(p^{1.5})$, we have $ \sqrt{p} \alpha_{xy}^2 = o(1) $. 
\end{remark}
In the next theorem, we state the asymptotic behavior of $ \text{ED}_n^{k}(\varGamma_{w}\mathbf{Z}) $ for each fixed permutation matrix $\varGamma_w \in \mathbb{S}_w$. Here, we use the italic gamma $\varGamma_w \in \mathbb{S}_w$ to differentiate from $\Gamma_i \in \mathbb{P}_{n+m}$.
	\begin{thm}
		\label{thm:lowall}
		For fixed $ \varGamma_w \in \mathbb{S}_w$,
		\begin{itemize}
			\item[(i)] Under Assumptions \ref{ass:1} and \ref{Ass:3},
			$$
			\emph{\text{ED}}_n^{k}( \varGamma_w \mathbf{Z}) \overset{p}{\rightarrow} \mu_{n, w},
			$$
			where $\mu_{n,w}$ is defined as
			\begin{multline*}
			\mu_{n,w} := \mu_{n}(\varGamma_w \mathbf{ Z})= (2\varphi( e_{xy} ) - \varphi(e_{x}) - \varphi (e_y)) \times \\ \left\lbrace  1
			-  \left( \frac{2m-1}{m(m-1)} + \frac{2n-1}{n(n-1)} \right)w
			+  \left( \frac{2}{mn}+ \frac{1}{n(n-1)} + \frac{1}{m(m-1)} \right) w^{2} \right\rbrace  .
			\end{multline*}
			\item[(ii)] Under Assumptions \ref{ass:1}, \ref{Ass:1}, \ref{Ass:4} and local alternative $H_{A_l}$, 
			$$
			\sqrt{p} \left( \emph{\text{ED}}_n^k( \varGamma_w \mathbf{Z}) -   \mu_n( \varGamma_w \mathbf{Z})\right) \overset{d}{\rightarrow} N\left( 0, \sigma^2_{n,w} \right).
			$$
			where $ \sigma^2_{n,w} $ is given as 
			\begin{align*}
			\sigma^2_{n,w} :=& \sigma_n(\varGamma_w \mathbf{ Z}) \\
			=&  \bigg\{ \frac{4}{nm} -4 \left( \frac{n+m}{n^2m^2} - \frac{n}{n^2(n-1)^2} - \frac{m}{m^2(m-1)^2} \right)w \\ 
			& \hspace{1cm}+ 4 \left(\frac{2}{n^2m^2} - \frac{1}{n^2(n-1)^2} - \frac{1}{m^2(m-1)^2} \right) w^{2} \bigg\} v_{xy} [\varphi^{(1)}(e_{xy})]^2  \\
			+ & \bigg\{ \frac{2}{n(n-1)} + 2\left( \frac{2n}{n^2m^2} - \frac{2n-1}{n^2(n-1)^2} - \frac{1}{m^2(m-1)^2} \right) w \\
			& \hspace{1cm} - 2\left(\frac{2}{m^2n^2} - \frac{1}{n^2(n-1)^2} - \frac{1}{m^2(m-1)^2} \right) w^{2}  \bigg\} v_x[ \varphi^{(1)}(e_{x})]^2 \\
			+ & \bigg\{ \frac{2}{m(m-1)} + 2\left( \frac{2m}{n^2m^2}  - \frac{1}{n^2(n-1)^2} - \frac{2m-1}{m^2(m-1)^2} \right) w \\
			& \hspace{1cm} -2 \left(\frac{2}{n^2m^2} - \frac{1}{m^2(m-1)^2} - \frac{1}{n^2(n-1)^2} \right) w^{2}  \bigg\} v_y[ \varphi^{(1)}(e_{y})]^2 .
			\end{align*}
		\end{itemize}
	\end{thm}
We use $W \sim \text{Hypergeometric}(n+m,m,n) $ to denote that $W$ follows the hypergeometric distribution, which describes the probability of $n$ draws from a union of two groups (one group has $m$ elements, the other has $n$ elements) such that $w$ of them are chosen from the group of size $m$. To be precise, $W$ has probability mass function
\begin{align*}
P(W=w) = \frac{\binom{m}{w}\binom{n}{m-w}}{\binom{n+m}{n}} \text{ for }  w \in \left\lbrace 0,1, \cdots, \min \{ n,m \} \right\rbrace .
\end{align*}
Then, the limiting distribution of $\text{ED}_n^{k}(\mathbf \Gamma \mathbf{Z})$ is derived in the following proposition.
\begin{proposition}
	\label{Lemmalow: asymHp}
	For $\mathbf \Gamma \sim \text{Uniform}(\mathbb P_{n+m})$, which is independent of the data, let 
	$$ 
	W:= N(\mathbf \Gamma) \sim \text{Hypergeometric}(n+m,m,n). 
	$$
	\begin{itemize}
		\item[(i)] Under Assumptions \ref{ass:1} and \ref{Ass:3},
		$$
		\emph{\text{ED}}_n^{k}(\mathbf \Gamma \mathbf{Z}) \overset{p}{\rightarrow} \mu_{n,W}.
		$$
		\item[(ii)] Under Assumptions \ref{ass:1}, \ref{Ass:1}, \ref{Ass:4} and local alternative $H_{A_l}$, 
		$$
		\sqrt{p} \left( \emph{\text{ED}}_n^k(\mathbf \Gamma \mathbf{Z}) -   \mu_{n,W}\right) \overset{d}{\rightarrow} N\left( 0, \sigma^2_{n,W} \right).
		$$
	\end{itemize}
\end{proposition}
In the above proposition, $ N\left( 0, \sigma^2_{n,W} \right) $ should be understood as a mixture of Gaussian with probability distribution 
\begin{align*}
P\left( N\left( 0, \sigma^2_{n,W} \right) \leq a \right) = \sum_{w=1}^{\min \{n,m\} } P(W=w) P(N\left( 0, \sigma^2_{n,w} \right) \leq a).
\end{align*}
Next, let $\Gamma_0$ corresponds to the identity permutation map, we present the power behavior of $ \text{ED}^k(\mathbf Z)$ when the critical values are obtained via permutations.
\begin{thm}
		\label{lemma:power1}
		Assume that $ 2 \varphi(e_{xy}) \geq  \varphi(e_{x})+  \varphi(e_{y}) $.
		\begin{itemize}
			\item[1,] \emph{[\textbf{Consistency}]} Suppose Assumptions \ref{ass:1} and \ref{Ass:3} hold. 
			\begin{itemize}
				\item[(i)] If the critical value is chosen as $ Q_{\widehat{R},1-\alpha} $. Let $n,m$ be large enough such that $
				n!m!/(n+m)! < 1 - \alpha \text{ if } m \neq n$ and $
				2 (n!)^2/(2n)! < 1 - \alpha \text{ if } m = n$. Then, we have
				$$
				\lim_{p \rightarrow \infty} P_{H_{A_c}} \left(\emph{\text{ED}}_n^k(\mathbf{Z}) > Q_{\widehat{R},1-\alpha} \right) = 1,
				$$
				which means that the asymptotic power of $\emph{\text{ED}}^k$ based permutation test is 1 as $p$ goes to infinity.
				\item[(ii)] If the critical value is chosen as $ Q_{\widetilde{R},1-\alpha} $.Then, we have
				\begin{align*}
				\lim_{p \rightarrow \infty} P_{H_{A_c}} \left(\emph{\text{ED}}_n^k(\mathbf{Z}) > Q_{\widetilde{R},1-\alpha} \right)  \geq  \left\{ 
				\begin{array}{ll}
				1 - \frac{S-1}{ \lfloor \alpha S \rfloor } \frac{n!m!}{(n+m)!}, & \text{if } n \neq m , \\
				1 - \frac{S-1}{ \lfloor \alpha S \rfloor } \frac{2(n!)^2}{(n+m)!}, & \text{if } n = m .
				\end{array}	\right.
				\end{align*}
			\end{itemize}
			
			\item[2,] \emph{[\textbf{Power Limit}]} Suppose Assumptions \ref{ass:1}, \ref{Ass:1}, \ref{Ass:4} hold. 
			\begin{itemize}
				\item[(i)] If the critical value is chosen as $ Q_{\widehat{R},1-\alpha} $. Then, we have
				\begin{align*}
				\lim_{p \rightarrow \infty}P_{H_{A_l}}\left(\emph{\text{ED}}_n^k(\mathbf{Z}) > Q_{\widehat{R},1-\alpha}\right)  = P(V(\Gamma_0) >  Q_{\widehat{T}, 1-\alpha} ) ,
				\end{align*}
				where 
				\begin{align*}
				\widehat{T}(t) := \frac{1}{(n+m)!} \sum\limits_{i=1}^{(n+m)!} \mathbb{I}_{ \left\lbrace V(\Gamma_{i})  \leq t \right\rbrace }
				\end{align*}
				and
				$$
				V (\Gamma_s)  = \sum\limits_{i=1}^{n} \sum\limits_{j=1}^{m} \Pi_{s,ij} b_{ij}  -   \sum\limits_{1 \leq i < j \leq n} \Pi_{s,ij} c_{ij} -  \sum\limits_{1 \leq i < j \leq m} \Pi_{s,ij} d_{ij}.
				$$
				\item[(ii)] If the critical value is chosen as $ Q_{\widetilde{R},1-\alpha} $.
				$$
				\lim_{p \rightarrow \infty}P_{H_{A_l}}\left(\emph{\text{ED}}_n^k(\mathbf{Z}) > Q_{\widetilde{R},1-\alpha}\right) = P\left(V(\Gamma_0) >  Q_{\widetilde{T},1-\alpha} \right),
				$$
				where 
				\begin{align*}
				\widetilde{T} :=\frac{1}{S} \left(\mathbb{I}_{ \left\lbrace V(\Gamma_0)  \leq t \right\rbrace } + \sum\limits_{i=1}^{S-1}  \mathbb{I}_{ \left\lbrace V(\mathbf \Gamma_{i} )  \leq t \right\rbrace }  \right).
				\end{align*}
			\end{itemize}
			
			\item[3,] \emph{[\textbf{Trivial Power}]} Suppose Assumptions \ref{ass:1}, \ref{Ass:1}, \ref{Ass:4} hold. Then, we have
			$$
			\lim_{p \rightarrow \infty}P_{H_{A_t}} \left( \emph{\text{ED}}_n^k(\mathbf{Z}) > c \right) \leq \alpha \text{ where } c = Q_{\widehat{R},1-\alpha} \text{ or } c= Q_{\widetilde{R},1-\alpha},
			$$
			which means that the asymptotic power of $\emph{\text{ED}}^k$ based permutation test is no more than the level $\alpha$ when $p$ goes to infinity.
		\end{itemize}
	\end{thm}
	\begin{remark}
		The above theorem and discussions in subsection \ref{subsec:local} indicate that
		\begin{itemize}
			\item[1,] $L^1$-norm can be more advantageous than $L^2$-norm, Gaussian kernel and Laplacian kernel when the dimension is high, since $L^1$-distance leads to high power provided that the summation of discrepancies between marginal univariate distributions is not so small, while $L^2$-norm, Gaussian kernel and Laplacian kernel would result in power loss when the total of marginal univariate mean and variance differences between $X$ and $Y$ is of order $o(\sqrt{p})$. Notice that the distributions of $X$ and $Y$ can differ in other aspects of the marginal distribution even if they have the same marginal univariate mean and variance.
			\item[2,] All the tests under examination are only capable of detecting the discrepancies of marginal distributions. If the two high dimensional distributions $F \neq G$, but $F_u = G_u$ for $u =1,2, \cdots, p$, then none of them have consistent power.
		\end{itemize}
	\end{remark}
	
	\subsection{High Dimensional Medium Sample Size (HDMSS)}
	In this subsection, the theories are developed under the high dimensional medium sample size setting (HDMSS), i.e., as $p \rightarrow \infty$, $n:=n(p) \rightarrow \infty$ at a slower rate compared to $p$ and $n/m = \rho$, where $\rho \in (0, \infty)$ is a fixed constant. Though the proofs are quite different, most results and phenomena under the HDLSS setting have their similar counterparts under the HDMSS setting. Now that we have $n,m$ growing to infinity, we need a stronger version of Assumptions \ref{Ass:3} and \ref{Ass:4}.
	\begin{assumption} \label{ass:2}
		$
		nm \alpha_{xy}^2 = o(1), n^2 \alpha_{x}^2 = o(1) \text{ and } m^2\alpha_{y}^2 = o(1).
		$
	\end{assumption}
	\begin{assumption} \label{ass:3}
		$
		\sqrt{nmp} \alpha_{xy}^2 = o(1), n\sqrt{p} \alpha_{x}^2 = o(1) \text{ and } m\sqrt{p}\alpha_{y}^2 = o(1).
		$
	\end{assumption}
	\begin{remark}
		Following Remark \ref{rmk:l1}, for $\kappa$-dependent stationary time series, $\alpha_{xy}^2 = O(\kappa /p)$. Thus, Assumptions \ref{ass:2} and \ref{ass:3} both require that $nm \kappa = o(p) $.
	\end{remark}
	
	To derive the asymptotic distribution under the HDMSS, we note that the leading term of $\text{ED}^k_n(\mathbf Z)$ is a martingale, the following assumption is used to ensure the conditional Lindeberg
	condition and the requirements on the conditional variance in classic martingale central limit theorem.
	\begin{assumption} \label{ass:4}
		For any $\Lambda_1, \Lambda_2, \Lambda_3, \Lambda_4 \in \{ X,Y \} $, suppose
		\begin{align*}
		& E \left[ \mathcal{K}^4(\Lambda_1, \Lambda_2')  \right] = o\left( n^2 \right),  \\
		&E \left[ \mathcal{K}^2(\Lambda_1, \Lambda_3'') \mathcal{K}^2(\Lambda_2', \Lambda_3'') \right] = o(n),  \\
		&E \left[  \mathcal{K}(\Lambda_1, \Lambda_3'') \mathcal{K}(\Lambda_1, \Lambda_4''')  \mathcal{K}(\Lambda_2', \Lambda_4''') \mathcal{K}(\Lambda_2', \Lambda_3'')  \right] = o(1),
		\end{align*}
		where $(\Lambda_1', \Lambda_2', \Lambda_3', \Lambda_4')$, $(\Lambda_1'', \Lambda_2'', \Lambda_3'', \Lambda_4'')$ and  $(\Lambda_1''', \Lambda_2''', \Lambda_3''', \Lambda_4''') $ are independent copies of $(\Lambda_1, \Lambda_2, \Lambda_3, \Lambda_4)$.
	\end{assumption}
	\begin{remark}
		For any function $\varphi$ and $\psi$, suppose $X$ and $Y$ are $\kappa$ dependent sequences, i.e., $x_u \perp x_{v}$ and $y_u \perp y_v$ if $|u-v| > \kappa$. If there exists some constant $C>0$ such that 
		\begin{align*}
		\max  \left\lbrace  \sup_{u} E\left[ \psi^4 (x_u, y_u) \right] , \sup_{u} E\left[ \psi^4 (x_u, x_u') \right],  \sup_{u} E \left[ \psi^4 (y_u, y_u') \right] \right\rbrace  \leq C.
		\end{align*}
		Then, for notational convenience, let 
		\begin{align*}
		\phi_{xy, u} = \psi (x_{u}, y_{u}) - E\left[\psi (x_{u}, y_{ju})|x_{u}\right] -  E\left[\psi (x_{u}, y_{u})|y_{u}\right] +  E\left[\psi (x_{u} , y_{u})\right],
		\end{align*}
		we see that $\sup_u E[\phi_{xy, u}^4] \leq 4^4C $ and thus $E [ \mathcal{K}^4(X, Y)  ]$ can be bounded as following
		\begin{align*}
		E \left[ \mathcal{K}^4(X, Y)  \right] 
		& = \frac{1}{p^2} \sum_{s=1}^p \sum_{t,u,v = s - 3 \kappa} ^{s+3\kappa} E[ \phi_{xy, s} \phi_{xy, t}\phi_{xy, u} \phi_{xy, v}]  = O\left(  \frac{\kappa^3}{p} \right).
		\end{align*}
		and similar results can be shown for $E [ \mathcal{K}^4(X, X')  ]$ and $E [ \mathcal{K}^4(Y, Y')  ]$. Thus, Assumption \ref{ass:4} is satisfied if $\kappa^3/p = o(1)$.
	\end{remark}
Let $\Phi$ denote the cdf of $N(0,1)$. We shall show that $\text{ED}^k (\varGamma_w \mathbf{ Z})$ converges uniformly with respect to $w$ under the HDMSS setting. 
		\begin{thm}\label{lem:per} For $ w =0, 1,2, \cdots, \min \{n,m \}$, fix $\varGamma_w \in \mathbb{S}_w$,
		\label{Lemma: asymHp}
		\begin{itemize}
			\item[(i)]Under Assumptions \ref{ass:1} and \ref{ass:2}, 
			\begin{align*}
			\sup_w \left| \emph{\text{ED}}_n^{k}(\varGamma_{w}\mathbf{Z}) - \mu_{n,w} \right|   = o_p(1).
			\end{align*}
			where $\mu_{n,w}$ is the same as that in Theorem \ref{thm:lowall}.
			\item[(ii)]  Under Assumptions \ref{ass:1}, \ref{ass:2}, \ref{ass:3}, \ref{ass:4} and local alternative $H_{A_l}$,
			$$
			 \sup_w \left| P\left( \sqrt{nmp} \left( \emph{\text{ED}}_n^k( \varGamma_{w} \mathbf{Z}) -  \mu_{n,w}  \right) \leq a \right) - \Phi \left(\frac{a}{\sqrt{nm\sigma^2_{n,w}}} \right) \right| = o(1)
			$$
			where $a$ is a fixed constant and $ \sigma^2_{n,w} $ is the same as in Theorem \ref{thm:lowall}. 
		\end{itemize}
	\end{thm}
	Then, the following theorem states the asymptotic distribution of $\text{ED}_n^{k}(\mathbf \Gamma \mathbf{Z})$.
	\begin{thm} \label{lem:rand}
		For $\mathbf \Gamma \sim \text{Uniform}(\mathbb P_{n+m})$, which is independent of the data,
		\begin{itemize}
			\item[(i)] Under Assumptions \ref{ass:1} and \ref{ass:2}, 
			\begin{align*}
			\emph{ED}_n^{k}(\mathbf \Gamma \mathbf{Z}) \overset{p}{\rightarrow}  0.
			\end{align*}
			\item[(ii)] Under Assumptions \ref{ass:1}, \ref{ass:2}, \ref{ass:3}, \ref{ass:4} and local alternative $H_{A_l}$,
			\begin{align*}
			\sqrt{nmp} \left( \begin{array}{l}
			\emph{ED}_n^k(\mathbf \Gamma \mathbf{Z})  \\
			\emph{ED}_n^k(\mathbf \Gamma' \mathbf{Z}) 
			\end{array}  \right) \overset{d}{\rightarrow} N \left(0,  
			\left( \begin{array}{cc}
			\sigma^2 &  0 \\
			0 & \sigma^2
			\end{array} \right) \right) 
			\end{align*}
			where $\mathbf \Gamma'$ is an independent copy of $\mathbf \Gamma $ and $\sigma^2$ is the asymptotic variance defined as
			\begin{align*}
			\sigma^2 := 4 v_{xy} [\varphi^{(1)}(e_{xy})]^2 + 2 \rho v_x[ \varphi^{(1)}(e_{x})]^2 + \frac{2}{\rho} v_y[ \varphi^{(1)}(e_{y})]^2.
			\end{align*}
		\end{itemize}	
	\end{thm}
We need the limiting distribution of  $( \text{ED}_n^k(\mathbf \Gamma \mathbf{Z}), \text{ED}_n^k(\mathbf \Gamma' \mathbf{Z}) )$ to show that the variance of randomization distribution go to 0, from which it follows that the randomization distribution converges in probability to the limit of its mean. Furthermore, we can show that the critical values are concentrating on some constants.
		\begin{corollary} \label{cor:main} Let $\bm{\Gamma}_1, \cdots, \bm{\Gamma}_S $ be i.i.d and  uniformly sampled from $\mathbb{P}_{n+m}$.
		\begin{itemize}
			\item[(i)] Under Assumptions \ref{ass:1} and \ref{ass:2}, as $n \wedge m \wedge p \wedge S \rightarrow  \infty$,
			\begin{align*}
			\widehat{R}(t) \overset{p}{\rightarrow}  \mathbb{I}_{\{ t\geq 0\} } \text{ and }
			\widetilde{R}(t)  \overset{p}{\rightarrow}  \mathbb{I}_{\{ t\geq 0\}} . 
			\end{align*}
			Consequently, we have $ Q_{\widehat{R},1-\alpha} \overset{p}{\rightarrow}  0 $ and $ Q_{\widetilde{R},1-\alpha} \overset{p}{\rightarrow}  0$.
			\item[(ii)] Under Assumptions \ref{ass:1}, \ref{ass:2}, \ref{ass:3}, \ref{ass:4} and local alternative $H_{A_l}$, as $n \wedge m \wedge p \wedge S \rightarrow  \infty$,
			\begin{align*}
			\frac{1}{(n+m)!} \sum\limits_{i=1}^{(n+m)!} \mathbb{I}_{ \left\lbrace \sqrt{nmp} \emph{ED}_n^k(\Gamma_{i} \mathbf{Z})  \leq t \right\rbrace } & \overset{p}{\rightarrow}  \Phi (t/\sigma), \\
			\frac{1}{S} \sum\limits_{i=1}^{S} \mathbb{I}_{ \left\lbrace \sqrt{nmp} \emph{ED}_n^k(\mathbf \Gamma_{i} \mathbf{Z})  \leq t \right\rbrace } & \overset{p}{\rightarrow}  \Phi (t/\sigma) 
			\end{align*} 
			Consequently,  we have $ \sqrt{nmp} Q_{\widehat{R}, 1 - \alpha} \overset{p}{\rightarrow} \sigma Q_{\Phi, 1-\alpha} $ and $ \sqrt{nmp} Q_{\widetilde{R}, 1 - \alpha} \overset{p}{\rightarrow} \sigma Q_{\Phi, 1- \alpha} $, where $\sigma^2$ is defined in Theorem \ref{lem:rand}.
		\end{itemize}
	\end{corollary}

	The power behavior of $\text{ED}_n^k(\mathbf{Z})$ w.r.t permutation test under the HDMSS is stated in the following theorem.
	\begin{thm}
		\label{lemma:power2}
		Assume that $ 2 \varphi(e_{xy}) \geq  \varphi(e_{x})+  \varphi(e_{y}) $. For any $c \in \{ Q_{\widehat{R},1-\alpha}, Q_{\widetilde{R},1-\alpha} \}$, the following holds.
		\begin{itemize}
			\item[1,] \emph{[\textbf{Consistency}]} Under Assumptions \ref{ass:1}, \ref{ass:2}. Then,  we have
			$$
			\lim_{p \rightarrow \infty} P_{H_{A_c}} \left(\emph{\text{ED}}_n^k(\mathbf{Z}) > c \right) = 1,
			$$
			which means that the asymptotic power of $\emph{\text{ED}}^k$ based permutation test is 1 as $p \wedge n \wedge m \rightarrow \infty$.
			\item[2,] \emph{[\textbf{Trivial Power}]} Under Assumptions \ref{ass:1}, \ref{ass:2}, \ref{ass:3}, \ref{ass:4}. Then, we have 
			$$
			\lim_{p \rightarrow \infty}P_{H_{A_l}}\left(\emph{\text{ED}}_n^k(\mathbf{Z}) > c\right) \leq  \alpha,
			$$
			Thus, we have the asymptotic power of $\emph{\text{ED}}^k$ based permutation test is no more than the level $\alpha$ when $p\wedge n \wedge m \rightarrow \infty$.
		\end{itemize}
	\end{thm}
	Comparing with Theorem \ref{lemma:power1}, the $ \text{ED}_n^k(\mathbf{Z}) $ based permutation test have trivial power under $H_{A_l}$ and the HDMSS setting. This is due to the interesting facts that $nm\sigma^2_{n,W}$ converges in probability to $\sigma^2$, which is also the limit of $nm\sigma^2_{n,0}$ as $n \rightarrow \infty$ and 
	\begin{align*}
	\cov \left( \text{ED}_n^k(\mathbf \Gamma \mathbf{Z}),
	\text{ED}_n^k(\mathbf \Gamma' \mathbf{Z})  \right) \rightarrow 0 \text{ as } n \rightarrow \infty,
	\end{align*}
	which ensures that the randomization distribution converges in probability to its mean limit.
	
		\section{Numerical Studies}
	\label{sec:twosample:sim}
	In this section, we consider several examples to demonstrate the finite sample performance of $\text{ED}^k$ based permutation test for different distance metrics. In our numerical comparison, we include the tests of \citet{li2018asymptotic} (denoted as JL) and \citet{biswas2014nonparametric} (denoted as BG) as these two were shown to have higher power over others in \citet{li2018asymptotic}. The critical values of JL test are determined by its asymptotic distribution, whereas  BG test is also implemented as a permutation test.
	\subsection{Performance on simulated data}
	In all our simulations, we set $\alpha = 0.05$ and perform 1000 Monte Carlo replications with 300 permutations for each test.  The first example is adopted from the simulation setting of \cite{li2018asymptotic} to study the size accuracy.
	\begin{example}
		\label{eg:twosample0}
		Generate samples as
		\begin{align*}
		& X = (V^{1/2} R V^{1/2})^{1/2} Z_1, \\
		& Y = (V^{1/2} R V^{1/2})^{1/2} Z_2,
		\end{align*}
		where $R = (r_{ij})_{i,j=1}^p$, $ r_{ij} = \rho^{|i-j|}$ and $\rho = 0.5 \text{ or }0.8$; $V$ is a diagonal matrix with $V_{ii}^{1/2}=1$ or uniformly drawn from (1,5). $Z_1$, $Z_2$ are i.i.d copies of $Z$ with
		\begin{align*}
		Z =(\underbrace{z_{1}, z_2, \cdots ,z_{p}}_{ \overset{iid}{\sim} N( 0, 1)})  \text{ or } Z =(\underbrace{z_{1}, z_2, \cdots ,z_{p}}_{ \overset{iid}{\sim} \text{Exponential}(1)}) - \mathbf{1}_p.
		\end{align*}
	\end{example}
	In Example \ref{eg:twosample0}, $X$ and $Y$ follow the same distribution and we consider cases that $n=m=50$ or $n=70, m=30$. From Table \ref{tab2}, we can see that all the tests have quite accurate size.
	\begin{table}[!htbp] \centering
		\caption{Size comparison from Example \ref{eg:twosample0} for $p=500$}
		\label{tab2}
		\scalebox{0.8}{
			\begin{tabular}{@{\extracolsep{5pt}} ccccccccccc}
				\\[-1.8ex]\hline \hline \\[-1.8ex]
				&$\rho$ & $V_{ii}^{1/2}$ & $n$ & $m$ & $ \text{ED}^{L^2\text{-norm}} $ & $ \text{ED}^{\text{Gaussian}} $ &$ \text{ED}^{\text{Laplacian}} $ & $ \text{ED}^{L^1\text{-norm}} $ & BG & JL   \\
				\hline \\[-1.8ex]
				\multirow{8}{*}{\rotatebox{90}{Normal}}& 0.5&1 & 50 & 50 & 0.06 &0.06& 0.058& 0.059& 0.053& 0.053  \\
				&0.5&1 & 70 & 30 & 0.07& 0.07& 0.068& 0.073& 0.047& 0.057 \\
				&0.5&Un(1,5) & 50 & 50 & 0.052& 0.052& 0.05& 0.051& 0.056& 0.057 \\
				&0.5&Un(1,5) & 70 & 30 & 0.059& 0.059& 0.061& 0.05& 0.049& 0.045 \\
				&0.8&1 & 50 & 50 & 0.053& 0.053& 0.052& 0.059& 0.054& 0.055 \\
				&0.8&1 & 70 & 30 & 0.045& 0.046& 0.046& 0.05& 0.052& 0.055\\
				&0.8&Un(1,5) & 50 & 50 & 0.045& 0.045& 0.049& 0.048& 0.054& 0.054 \\
				&0.8&Un(1,5) & 70 & 30 & 0.05& 0.05& 0.049& 0.046& 0.051& 0.051 \\
				\hline \\[-1.8ex] 	
				\multirow{8}{*}{\rotatebox{90}{Exponential}}& 0.5&1 & 50 & 50 & 0.06& 0.06& 0.058& 0.059& 0.053& 0.053 \\
				&0.5&1 & 70 & 30 & 0.063& 0.063& 0.063& 0.058& 0.048& 0.053\\
				&0.5&Un(1,5) & 50 & 50 & 0.057 & 0.057 & 0.058 &0.055& 0.049& 0.06\\
				&0.5&Un(1,5) & 70 & 30 &0.056& 0.056& 0.06& 0.058& 0.059& 0.058 \\
				&0.8&1 & 50 & 50 & 0.054& 0.054& 0.051& 0.047& 0.065& 0.062  \\
				&0.8&1 & 70 & 30 & 0.061& 0.061& 0.062& 0.065& 0.057& 0.06 \\
				&0.8&Un(1,5) & 50 & 50 &0.051& 0.05& 0.052& 0.046& 0.045& 0.057 \\
				&0.8&Un(1,5) & 70 & 30 &0.062& 0.062& 0.062& 0.062& 0.06& 0.064 \\
				\hline \\[-1.8ex] 	
		\end{tabular} }
	\end{table}
	To compare the power, we first use an example from \cite{li2018asymptotic}, which include the situation when $X$ and $Y$ only differ in their means or only differ in their covariance matrices or differ in both, where $\beta \in [0,1]$ is the percentage of the $p$ components that differ in their distributions.
	
		\begin{figure}[ht]
		\centering
		\begin{tikzpicture}[scale=0.5]
		\begin{axis}[
		name = tx1,
		title= {Example \ref{eg:JL} (i)} ,
		xlabel={$ \beta $},
		ylabel={ Power },
		ymin=0, ymax=1,
		grid = major,
		]
		\addlegendentry{$\text{ED}^{L^2\text{-norm}}$}
		\addplot [mark=o ,loosely dotted, every mark/.append style={solid}] table[x index=0, y index=1, col sep=comma] {1.dat};
		
		\addlegendentry{$\text{ED}^{\text{Gaussian}}$}
		\addplot [mark=square, dashed, every mark/.append style={solid}] table[x index=0, y index=2, col sep=comma] {1.dat};
		
		\addlegendentry{$\text{ED}^{\text{Laplacian}}$}
		\addplot [mark=triangle, loosely dashed, every mark/.append style={solid}] table[x index=0, y index=3, col sep=comma] {1.dat};
		
		\addlegendentry{$\text{ED}^{L^{1}}$}
		\addplot  table[x index=0, y index=4, col sep=comma] {1.dat};
		
		\addlegendentry{$\text{BG}$}
		\addplot [mark=diamond, densely dotted, every mark/.append style={solid}] table[x index=0, y index=5, col sep=comma] {1.dat};
		
		\addlegendentry{$\text{JL}$}
		\addplot [mark=x, dotted, every mark/.append style={solid}] table[x index=0, y index=6, col sep=comma] {1.dat};
		
		\legend{}
		\end{axis}

		\begin{axis}[
		name = tx2,
		at={(tx1.south east)},
		xshift=1cm,
		title= {Example \ref{eg:JL} (ii)} ,
		xlabel={$\beta $},
		ymin=0, ymax=1,
		grid = major,
		legend pos=outer north east,
		]
		\addlegendentry{$\text{ED}^{L^2\text{-norm}}$}
		\addplot [mark=o,loosely dotted, every mark/.append style={solid}] table[x index=0, y index=1, col sep=comma] {2.dat};
		
		\addlegendentry{$\text{ED}^{\text{Gaussian}}$}
		\addplot [mark=square, dashed, every mark/.append style={solid}] table[x index=0, y index=2, col sep=comma] {2.dat};
		
		\addlegendentry{$\text{ED}^{\text{Laplacian}}$}
		\addplot [mark=triangle, loosely dashed, every mark/.append style={solid}] table[x index=0, y index=3, col sep=comma] {2.dat};
		
		\addlegendentry{$\text{ED}^{L^{1}\text{-norm}}$}
		\addplot  table[x index=0, y index=4, col sep=comma] {2.dat};
		
		\addlegendentry{$\text{BG}$}
		\addplot [mark=diamond, densely dotted, every mark/.append style={solid}] table[x index=0, y index=5, col sep=comma] {2.dat};
		
		\addlegendentry{$\text{JL}$}
		\addplot [mark=x, dotted, every mark/.append style={solid}] table[x index=0, y index=6, col sep=comma] {2.dat};
		\legend{}
		\end{axis}
		
		\begin{axis}[
		at={(tx2.south east)},
		xshift=1cm,
		title= {Example \ref{eg:JL} (iii)} ,
		xlabel={$\beta $},
		ymin=0, ymax=1,
		grid = major,
		legend style={at={(-0.6,-0.2)},anchor=north, legend columns=-1}
		]
		\addlegendentry{$\text{ED}^{L^2\text{-norm}}$}
		\addplot [mark=o,loosely dotted, every mark/.append style={solid}] table[x index=0, y index=1, col sep=comma] {3.dat};
		
		\addlegendentry{$\text{ED}^{\text{Gaussian}}$}
		\addplot [mark=square, dashed, every mark/.append style={solid}] table[x index=0, y index=2, col sep=comma] {3.dat};
		
		\addlegendentry{$\text{ED}^{\text{Laplacian}}$}
		\addplot [mark=triangle, loosely dashed, every mark/.append style={solid}] table[x index=0, y index=3, col sep=comma] {3.dat};
		
		\addlegendentry{$\text{ED}^{L^{1}\text{-norm}}$}
		\addplot  table[x index=0, y index=4, col sep=comma] {3.dat};
		
		\addlegendentry{$\text{BG}$}
		\addplot [mark=diamond, densely dotted, every mark/.append style={solid}] table[x index=0, y index=5, col sep=comma] {3.dat};
		
		\addlegendentry{$\text{JL}$}
		\addplot [mark=x, dotted, every mark/.append style={solid}] table[x index=0, y index=6, col sep=comma] {3.dat};
		\legend{}
		\end{axis}
		
		\begin{axis}[
		name = ax1,
		at={(tx1.below south west)},
		yshift=-6.7cm,
		title= {Example \ref{eg:JL} (i)} ,
		xlabel={$ \beta $},
		ylabel={ Power },
		ymin=0, ymax=1,
		grid = major,
		]
		\addlegendentry{$\text{ED}^{L^2\text{-norm}}$}
		\addplot [mark=o,loosely dotted, every mark/.append style={solid}] table[x index=0, y index=1, col sep=comma] {e1.dat};
		
		\addlegendentry{$\text{ED}^{text{Gaussian}}$}
		\addplot [mark=square, dashed, every mark/.append style={solid}] table[x index=0, y index=2, col sep=comma] {e1.dat};
		
		\addlegendentry{$\text{ED}^{text{Laplacian}}$}
		\addplot [mark=triangle, loosely dashed, every mark/.append style={solid}] table[x index=0, y index=3, col sep=comma] {e1.dat};
		
		\addlegendentry{$\text{ED}^{L^{1}}$}
		\addplot table[x index=0, y index=4, col sep=comma] {e1.dat};
		
		\addlegendentry{$\text{BG}$}
		\addplot [mark=diamond, densely dotted, every mark/.append style={solid}] table[x index=0, y index=5, col sep=comma] {e1.dat};
		
		\addlegendentry{$\text{JL}$}
		\addplot [mark=x, dotted, every mark/.append style={solid}] table[x index=0, y index=6, col sep=comma] {e1.dat};
		
		\legend{}
		\end{axis}

		\begin{axis}[
		name = ax2,
		at={(ax1.south east)},
		xshift=1cm,
		title= {Example \ref{eg:JL} (ii)} ,
		xlabel={$\beta $},
		ymin=0, ymax=1,
		grid = major,
		legend pos=outer north east,
		]
		\addlegendentry{$\text{ED}^{L^2\text{-norm}}$}
		\addplot [mark=o,loosely dotted, every mark/.append style={solid}]  table[x index=0, y index=1, col sep=comma] {e2.dat};
		
		\addlegendentry{$\text{ED}^{\text{Gaussian}}$}
		\addplot [mark=square, dashed, every mark/.append style={solid}] table[x index=0, y index=2, col sep=comma] {e2.dat};
		
		\addlegendentry{$\text{ED}^{\text{Laplacian}}$}
		\addplot [mark=triangle, loosely dashed, every mark/.append style={solid}] table[x index=0, y index=3, col sep=comma] {e2.dat};
		
		\addlegendentry{$\text{ED}^{L^{1}\text{-norm}}$}
		\addplot  table[x index=0, y index=4, col sep=comma] {e2.dat};
		
		\addlegendentry{$\text{BG}$}
		\addplot [mark=diamond, densely dotted, every mark/.append style={solid}]  table[x index=0, y index=5, col sep=comma] {e2.dat};
		
		\addlegendentry{$\text{JL}$}
		\addplot [mark=x, dotted, every mark/.append style={solid}] table[x index=0, y index=6, col sep=comma] {e2.dat};
		\legend{}
		\end{axis}
		
		\begin{axis}[
		at={(ax2.south east)},
		xshift=1cm,
		title= {Example \ref{eg:JL} (iii)} ,
		xlabel={$\beta $},
		ymin=0, ymax=1,
		grid = major,
		legend style={at={(-0.6,-0.2)},anchor=north, legend columns=-1}
		]
		\addlegendentry{$\text{ED}^{L^2\text{-norm}}$}
		\addplot[mark=o,loosely dotted, every mark/.append style={solid}]  table[x index=0, y index=1, col sep=comma] {e3.dat};
		
		\addlegendentry{$\text{ED}^{\text{Gaussian}}$}
		\addplot [mark=square, dashed, every mark/.append style={solid}] table[x index=0, y index=2, col sep=comma] {e3.dat};
		
		\addlegendentry{$\text{ED}^{\text{Laplacian}}$}
		\addplot [mark=triangle, loosely dashed, every mark/.append style={solid}] table[x index=0, y index=3, col sep=comma] {e3.dat};
		
		\addlegendentry{$\text{ED}^{L^{1}\text{-norm}}$}
		\addplot   table[x index=0, y index=4, col sep=comma] {e3.dat};
		
		\addlegendentry{$\text{BG}$}
		\addplot [mark=diamond, densely dotted, every mark/.append style={solid}]  table[x index=0, y index=5, col sep=comma] {e3.dat};
		
		\addlegendentry{$\text{JL}$}
		\addplot [mark=x, dotted, every mark/.append style={solid}] table[x index=0, y index=6, col sep=comma] {e3.dat};
		\end{axis}
		\end{tikzpicture}
		\caption{Power comparison for example \ref{eg:JL} and $n=70$, $m=30$, $p=500$, where in the top 3 figures $Z_1, Z_2$ are generated from normal distribution and in the bottom 3 figures, $Z_1, Z_2$ are generated from exponential distribution.}
		\label{figuretest}
	\end{figure}
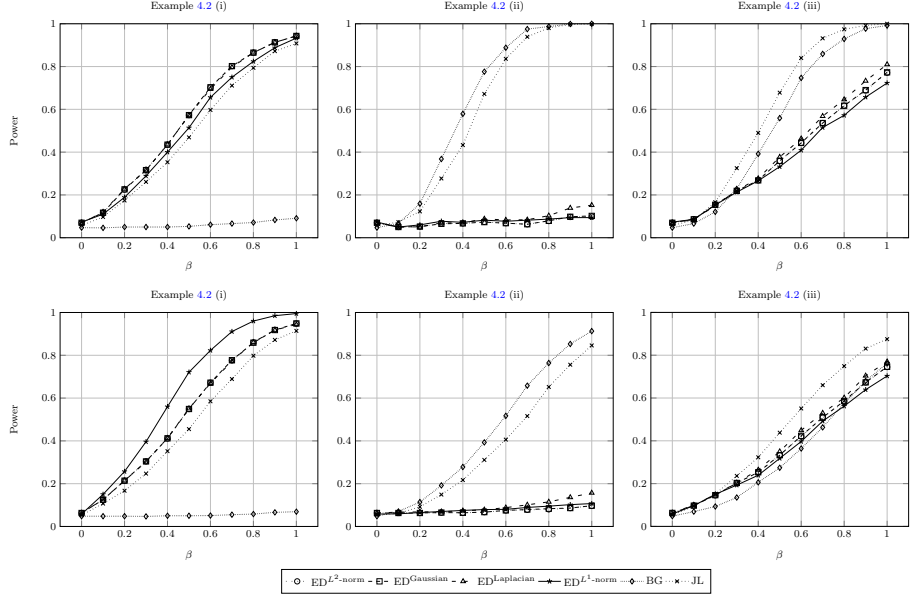
	
	\begin{example}
		\label{eg:JL}
		Let $R, V, Z_1, Z_2$ be defined the same as in Example \ref{eg:twosample0} and we choose $\rho = 0.5$ here. Generate samples as
		\begin{itemize}
			\item[(i)]
			\begin{align*}
			& X = (V^{1/2} R V^{1/2})^{1/2} Z_1, \\
			& Y =( 0.125 \times \mathbf{1}_{\beta p}, \mathbf{0}_{(1-\beta)p} ) + (V^{1/2} R V^{1/2})^{1/2} Z_2.
			\end{align*}
			\item[(ii)] Let $V^{*}$ be a diagonal matrix with $V^{*1/2}_{ii} =1.05$ for $i =1, 2, \cdots, \beta p$ and $V^{*1/2}_{ii} =1$ for $i =\beta p+1, \cdots, \beta p$.
			\begin{align*}
			& X = (V^{1/2} R V^{1/2})^{1/2} Z_1, \\
			& Y = (V^{*1/2} R V^{*1/2})^{1/2} Z_2.
			\end{align*}
			\item[(iii)] Let $V^{*1/2}_{ii} =1.04$ for $i =1, 2, \cdots, \beta p$ and $V^{*1/2}_{ii} =1$ for $i =\beta p+1, \cdots, \beta p$.
			\begin{align*}
			& X = (V^{1/2} R V^{1/2})^{1/2} Z_1, \\
			& Y =( 0.1 \times \mathbf{1}_{\beta p}, \mathbf{0}_{(1-\beta)p} ) + (V^{*1/2} R V^{*1/2})^{1/2} Z_2.
			\end{align*}
		\end{itemize}
	\end{example}
	From Figure \ref{figuretest}, we can see that (1) when there is a small difference in the means, $\text{ED}^k$-based tests and JL perform similarly, while BG barely show any power.
	(2) when there is a small difference in the scales, JL and BG are consistent and $\text{ED}^k$-based tests have very little power. Similar phenomenon by \citet{li2018asymptotic} were also observed, i.e., $\text{ED}^k$ based permutation test is not sensitive to small scale differences and the method proposed by \citet{li2018asymptotic} and \citet{biswas2014nonparametric} have dominant power in this case. Note that there is a tuning parameter involved in JL test and its choice could have a big impact on the size and power; results not shown. (3) when there are differences for both the means and scales, all the tests performs comparably.
	
	Next, Example \ref{eg:twosample1} examines the situation when $ X $ and $Y$ have the same marginal univariate mean and variance, but different marginal univariate distributions.
	
	\begin{example}
		\label{eg:twosample1}
		Generate samples as
		\begin{itemize}
			\item[(i)] Let Rademacher(0.5) be the Rademacher distribution with success probability 0.5, e.g. $P(y_{iu}= - 1) = P(y_{iu}=1) = 0.5$.
			\begin{align*}
			X &=(x_{1},...,x_{p}) \overset{iid}{\sim} N( 0, 1),  \\
			Y &= (\underbrace{y_{1}, y_{2}, \cdots, y_{\beta p}}_{\overset{iid}{\sim} \text{Rademacher}(0.5)}, \underbrace{y_{\beta p + 1}, y_{\beta p+2} \cdots,  y_{p}}_{\overset{iid}{\sim} N(0,1)}).
			\end{align*}
			\item[(ii)]
			\begin{align*}
			X &=(x_{1},...,x_{p}) \overset{iid}{\sim} N(0, 1),  \\
			Y &= (\underbrace{y_{1}, y_{2}, \cdots, y_{\beta p}}_{\overset{iid}{\sim} \text{Uniform}(- \sqrt{3}, \sqrt{3})}, \underbrace{y_{\beta p + 1}, y_{\beta p+2} \cdots,  y_{p}}_{\overset{iid}{\sim} N(0,1)}).
			\end{align*}
		\end{itemize}
	\end{example}
	From Figure \ref{figure2}, we see that only $\text{ED}^{L^{1}\text{-norm}}$ based permutation test has power growing as $\beta$ elevates ($p$ fixed) or $p$ increases ($\beta$ fixed). This phenomenon matches with our theories, which indicate that $L^2\text{-norm}$, Gaussian and Laplacian kernel can detect only marginal mean and variance differences. For $\text{ED}^{L^{1}\text{-norm}}$ based permutation test, the power is growing more rapidly for Example \ref{eg:twosample1} (i) than Example \ref{eg:twosample1} (ii), which might suggest that $L^1$-distance is more sensitive for the difference between continuous and discrete distributions.
	It is also apparent that the JL and BG tests show little power in this example.
	The next example examines the case where $X$ and $Y$ have the same marginal univariate distributions.
	\begin{figure}[ht]
		\centering
		\begin{tikzpicture}[scale = 0.6]
		\begin{axis}[
		name = ax1,
		title= {Example \ref{eg:twosample1} (i)} ,
		xlabel={$ \beta $},
		ylabel={ Power },
		ymin=0, ymax=1,
		grid = major,
		]
		\addlegendentry{$\text{ED}^{L^2\text{-norm}}$}
		\addplot [mark=o,loosely dotted, every mark/.append style={solid}]  table[x index=0, y index=1, col sep=comma] {exp11.dat};
		
		\addlegendentry{$\text{ED}^{text{Gaussian}}$}
		\addplot [mark=square, dashed, every mark/.append style={solid}]  table[x index=0, y index=2, col sep=comma] {exp11.dat};
		
		\addlegendentry{$\text{ED}^{text{Laplacian}}$}
		\addplot  [mark=triangle, loosely dashed, every mark/.append style={solid}]  table[x index=0, y index=3, col sep=comma] {exp11.dat};
		
		\addlegendentry{$\text{ED}^{L^{1}}$}
		\addplot table[x index=0, y index=4, col sep=comma] {exp11.dat};
		
		\addlegendentry{$\text{BG}$}
		\addplot [mark=diamond, densely dotted, every mark/.append style={solid}]   table[x index=0, y index=5, col sep=comma] {exp11.dat};
		
		\addlegendentry{$\text{JL}$}
		\addplot [mark=x, dotted, every mark/.append style={solid}] table[x index=0, y index=6, col sep=comma] {exp11.dat};
		
		\legend{}
		\end{axis}
		
		\begin{axis}[
		at={(ax1.south east)},
		xshift=1cm,
		title= {Example \ref{eg:twosample1} (ii)} ,
		xlabel={$\beta $},
		ymin=0, ymax=1,
		grid = major,
		legend pos=outer north east,
		]
		\addlegendentry{$\text{ED}^{L^2\text{-norm}}$}
		\addplot [mark=o,loosely dotted, every mark/.append style={solid}]  table[x index=0, y index=1, col sep=comma] {exp12.dat};
		
		\addlegendentry{$\text{ED}^{\text{Gaussian}}$}
		\addplot [mark=square, dashed, every mark/.append style={solid}]  table[x index=0, y index=2, col sep=comma] {exp12.dat};
		
		\addlegendentry{$\text{ED}^{\text{Laplacian}}$}
		\addplot  [mark=triangle, loosely dashed, every mark/.append style={solid}]  table[x index=0, y index=3, col sep=comma] {exp12.dat};
		
		\addlegendentry{$\text{ED}^{L^{1}\text{-norm}}$}
		\addplot  table[x index=0, y index=4, col sep=comma] {exp12.dat};
		
		\addlegendentry{$\text{BG}$}
		\addplot [mark=diamond, densely dotted, every mark/.append style={solid}]  table[x index=0, y index=5, col sep=comma] {exp12.dat};
		
		\addlegendentry{$\text{JL}$}
		\addplot [mark=x, dotted, every mark/.append style={solid}] table[x index=0, y index=6, col sep=comma] {exp12.dat};
		\end{axis}
		
		\begin{axis}[
		name = bx1,
		at={(ax1.below south west)},
		yshift=-6.3cm,
		xlabel={$p $},
		ylabel={ Power },
		ymin=0, ymax=1,
		grid = major,
		legend pos=outer north east,
		]
		\addlegendentry{$\text{ED}^{L^2\text{-norm}}$}
		\addplot [mark=o,loosely dotted, every mark/.append style={solid}]  table[x index=0, y index=1, col sep=space] {revise11.txt};
		
		\addlegendentry{$\text{ED}^{\text{Gaussian}}$}
		\addplot [mark=square, dashed, every mark/.append style={solid}]  table[x index=0, y index=2, col sep=space] {revise11.txt};
		
		\addlegendentry{$\text{ED}^{\text{Laplacian}}$}
		\addplot  [mark=triangle, loosely dashed, every mark/.append style={solid}]  table[x index=0, y index=3, col sep=space] {revise11.txt};
		
		\addlegendentry{$\text{ED}^{L^{1}\text{-norm}}$}
		\addplot  table[x index=0, y index=4, col sep=space] {revise11.txt};
		
		\addlegendentry{$\text{BG}$}
		\addplot [mark=diamond, densely dotted, every mark/.append style={solid}]  table[x index=0, y index=5, col sep=space] {revise11.txt};
		
		\addlegendentry{$\text{JL}$}
		\addplot [mark=x, dotted, every mark/.append style={solid}] table[x index=0, y index=6, col sep=space] {revise11.txt};
		\legend{}
		\end{axis}
		
		\begin{axis}[
		name = bx2,
		at={(bx1.south east)},
		xshift=1cm,
		xlabel={$p$},
		ymin=0, ymax=1,
		grid = major,
		legend pos=outer north east,
		]
		\addlegendentry{$\text{ED}^{L^2\text{-norm}}$}
		\addplot [mark=o,loosely dotted, every mark/.append style={solid}]  table[x index=0, y index=1, col sep=space] {revise12.txt};
		
		\addlegendentry{$\text{ED}^{\text{Gaussian}}$}
		\addplot [mark=square, dashed, every mark/.append style={solid}]  table[x index=0, y index=2, col sep=space] {revise12.txt};
		
		\addlegendentry{$\text{ED}^{\text{Laplacian}}$}
		\addplot  [mark=triangle, loosely dashed, every mark/.append style={solid}]  table[x index=0, y index=3, col sep=space] {revise12.txt};
		
		\addlegendentry{$\text{ED}^{L^{1}\text{-norm}}$}
		\addplot  table[x index=0, y index=4, col sep=space] {revise12.txt};
		
		\addlegendentry{$\text{BG}$}
		\addplot [mark=diamond, densely dotted, every mark/.append style={solid}]  table[x index=0, y index=5, col sep=space] {revise12.txt};
		
		\addlegendentry{$\text{JL}$}
		\addplot [mark=x, dotted, every mark/.append style={solid}] table[x index=0, y index=6, col sep=space] {revise12.txt};
		\legend{}
		\end{axis}
		
		\end{tikzpicture}
		\caption{Power comparison for Example \ref{eg:twosample1} and $n=70$, $m=30$. For the top two figures, the dimension $p$ is equal to 500 and we plot the power as $\beta$ ranges from $0$ to 1. For the bottom two figures, $\beta$ is fixed to be $0.5$ and the power is plotted with respect to $p$.}
		\label{figure2}
	\end{figure}
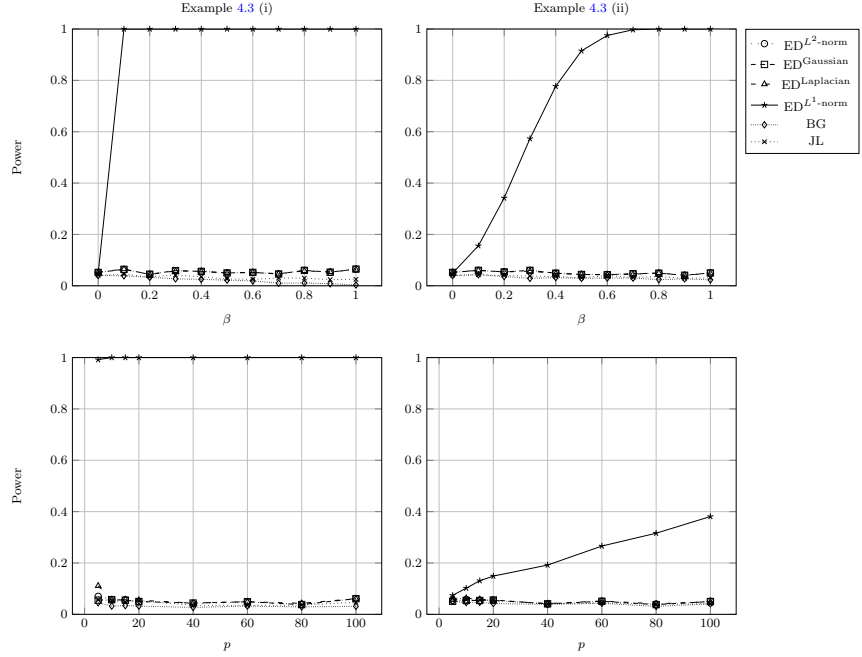

	\begin{example}
		\label{eg:twosample2}
		Generate samples as
		\begin{itemize}
			\item[(i)] Let $(y_{1}', y_{2}', \cdots , y_{\beta p/2}') \overset{iid}{\sim} \text{Bernoulli}(0.5)$
			\begin{align*}
			X &=(x_{1},...,x_{p}) \overset{iid}{\sim}  \text{Bernoulli}(0.5),  \\
			Y &= \big(y_{1}', \mathbb{I}_{ \{ y_{1}'=1 \} }, y_{2}', \mathbb{I}_{ \{ y_{2}'=1 \} } \cdots , y_{\beta p/2}', \mathbb{I}_{ \{y_{\beta p/2}'=1  \}}, \underbrace{y_1, y_2,  \cdots, y_{(1-\beta) p}}_{\overset{iid}{\sim} \text{Bernoulli}(0.5)}\big).
			\end{align*}
			\item[(ii)] Let $(y_{1}', y_{2}', \cdots , y_{\beta p/3}') \overset{iid}{\sim} \text{Bernoulli}(0.5)$ and $ (y_{1}'', y_{2}'', \cdots , y_{\beta p/3}'') \overset{iid}{\sim} \text{Bernoulli}(0.5) $
			\begin{align*}
			X &=(x_{1},...,x_{p}) \overset{iid}{\sim}  \text{Bernoulli}(0.5),  \\
			Y &= \big(y_{1}', y_{1}'', \mathbb{I}_{ \{ y_{1}'=y_{1}'' \} }, \cdots , y_{\beta p/3}', y_{\beta p/3}'', \mathbb{I}_{ \{y_{\beta p/3}'=y_{\beta p/3}''  \}}, \underbrace{y_1, y_2,  \cdots, y_{(1-\beta) p}}_{\overset{iid}{\sim} \text{Bernoulli}(0.5)}\big).
			\end{align*}
		\end{itemize}
	\end{example}
	Notice that in Example \ref{eg:twosample2} (i) $X$, $Y$ have the same marginal univariate distribution, but different marginal bivariate distributions and in Example \ref{eg:twosample2} (ii) $X$, $Y$ have the same marginal bivariate distribution, but different joint distribution. Theorem \ref{lemma:power1} (ii) and Theorem \ref{lemma:power2} (ii) both provide insights that $L^2$-norm, $L^1$-norm, Gaussian or Laplacian kernel based tests all suffer substantial power loss under Example \ref{eg:twosample2} (i). On the other hand, Theorem \ref{lemma:power1} (iii) suggests us that since Example \ref{eg:twosample2} (ii) belong to class $H_{A_t}$, all these tests have trivial power. The simulation results of Example \ref{eg:twosample2} are in Figure \ref{figure3} and they again corroborate our theoretical findings.
	
	\begin{figure}[ht]
		\centering
		\begin{tikzpicture}[scale = 0.6]
		\begin{axis}[
		name = ax1,
		title= {Example \ref{eg:twosample2} (i)} ,
		xlabel={$ \beta $},
		ylabel={Power},
		ymin=0, ymax=0.3,
		grid = major,
		]
		\addlegendentry{$\text{ED}^{L^2\text{-norm}}$}
		\addplot [mark=o,loosely dotted, every mark/.append style={solid}] table[x index=0, y index=1, col sep=comma] {exp22.dat};
		
		\addlegendentry{$\text{ED}^{text{Gaussian}}$}
		\addplot [mark=square, dashed, every mark/.append style={solid}] table[x index=0, y index=2, col sep=comma] {exp22.dat};
		
		\addlegendentry{$\text{ED}^{text{Laplacian}}$}
		\addplot [mark=triangle, loosely dashed, every mark/.append style={solid}] table[x index=0, y index=3, col sep=comma] {exp22.dat};
		
		\addlegendentry{$\text{ED}^{L^{1}}$}
		\addplot table[x index=0, y index=4, col sep=comma] {exp22.dat};
		
		\addlegendentry{$\text{BG}$}
		\addplot [mark=diamond, densely dotted, every mark/.append style={solid}]  table[x index=0, y index=5, col sep=comma] {exp22.dat};
		
		\addlegendentry{$\text{JL}$}
		\addplot [mark=x, dotted, every mark/.append style={solid}] table[x index=0, y index=6, col sep=comma] {exp22.dat};
		
		\legend{}
		\end{axis}
		
		\begin{axis}[
		at={(ax1.south east)},
		xshift=1cm,
		title= {Example \ref{eg:twosample2} (ii)} ,
		xlabel={$ \beta $},
		ymin=0, ymax=0.3,
		grid = major,
		legend pos=outer north east,
		]
		\addlegendentry{$\text{ED}^{L^2\text{-norm}}$}
		\addplot [mark=o,loosely dotted, every mark/.append style={solid}] table[x index=0, y index=1, col sep=comma] {exp21.dat};
		
		\addlegendentry{$\text{ED}^{\text{Gaussian}}$}
		\addplot [mark=square, dashed, every mark/.append style={solid}] table[x index=0, y index=2, col sep=comma] {exp21.dat};
		
		\addlegendentry{$\text{ED}^{\text{Laplacian}}$}
		\addplot [mark=triangle, loosely dashed, every mark/.append style={solid}] table[x index=0, y index=3, col sep=comma] {exp21.dat};
		
		\addlegendentry{$\text{ED}^{L^{1}\text{-norm}}$}
		\addplot table[x index=0, y index=4, col sep=comma] {exp21.dat};
		
		\addlegendentry{$\text{BG}$}
		\addplot [mark=diamond, densely dotted, every mark/.append style={solid}]  table[x index=0, y index=5, col sep=comma] {exp21.dat};
		
		\addlegendentry{$\text{JL}$}
		\addplot [mark=x, dotted, every mark/.append style={solid}] table[x index=0, y index=6, col sep=comma] {exp21.dat};
		\end{axis}
		
		\begin{axis}[
		name = bx1,
		at={(ax1.below south west)},
		yshift=-6.3cm,
		xlabel={$p$},
		ylabel={Power},
		ymin=0, ymax=1,
		grid = major,
		legend pos=outer north east,
		]
		\addlegendentry{$\text{ED}^{L^2\text{-norm}}$}
		\addplot [mark=o,loosely dotted, every mark/.append style={solid}]  table[x index=0, y index=1, col sep=space] {revise21.txt};
		
		\addlegendentry{$\text{ED}^{\text{Gaussian}}$}
		\addplot [mark=square, dashed, every mark/.append style={solid}]  table[x index=0, y index=2, col sep=space] {revise21.txt};
		
		\addlegendentry{$\text{ED}^{\text{Laplacian}}$}
		\addplot  [mark=triangle, loosely dashed, every mark/.append style={solid}]  table[x index=0, y index=3, col sep=space] {revise21.txt};
		
		\addlegendentry{$\text{ED}^{L^{1}\text{-norm}}$}
		\addplot  table[x index=0, y index=4, col sep=space] {revise21.txt};
		
		\addlegendentry{$\text{BG}$}
		\addplot [mark=diamond, densely dotted, every mark/.append style={solid}]  table[x index=0, y index=5, col sep=space] {revise21.txt};
		
		\addlegendentry{$\text{JL}$}
		\addplot [mark=x, dotted, every mark/.append style={solid}] table[x index=0, y index=6, col sep=space] {revise21.txt};
		\legend{}
		\end{axis}
		
		\begin{axis}[
		name = bx2,
		at={(bx1.south east)},
		xshift=1cm,
		xlabel={$p$},
		ymin=0, ymax=1,
		grid = major,
		legend pos=outer north east,
		]
		\addlegendentry{$\text{ED}^{L^2\text{-norm}}$}
		\addplot [mark=o,loosely dotted, every mark/.append style={solid}]  table[x index=0, y index=1, col sep=space] {revise22.txt};
		
		\addlegendentry{$\text{ED}^{\text{Gaussian}}$}
		\addplot [mark=square, dashed, every mark/.append style={solid}]  table[x index=0, y index=2, col sep=space] {revise22.txt};
		
		\addlegendentry{$\text{ED}^{\text{Laplacian}}$}
		\addplot  [mark=triangle, loosely dashed, every mark/.append style={solid}]  table[x index=0, y index=3, col sep=space] {revise22.txt};
		
		\addlegendentry{$\text{ED}^{L^{1}\text{-norm}}$}
		\addplot  table[x index=0, y index=4, col sep=space] {revise22.txt};
		
		\addlegendentry{$\text{BG}$}
		\addplot [mark=diamond, densely dotted, every mark/.append style={solid}]  table[x index=0, y index=5, col sep=space] {revise22.txt};
		
		\addlegendentry{$\text{JL}$}
		\addplot [mark=x, dotted, every mark/.append style={solid}] table[x index=0, y index=6, col sep=space] {revise22.txt};
		\legend{}
		\end{axis}
		
		\end{tikzpicture}
		\caption{Power comparison for Example \ref{eg:twosample2} and $n=70$, $m=30$. For the top two figures, the dimension $p$ is equal to 500 and we plot the power as $\beta$ ranges from $0$ to 1. For the bottom two figures, $\beta$ is fixed to be $1$ and the power is plotted with respect to $p$.}
		\label{figure3}
	\end{figure}
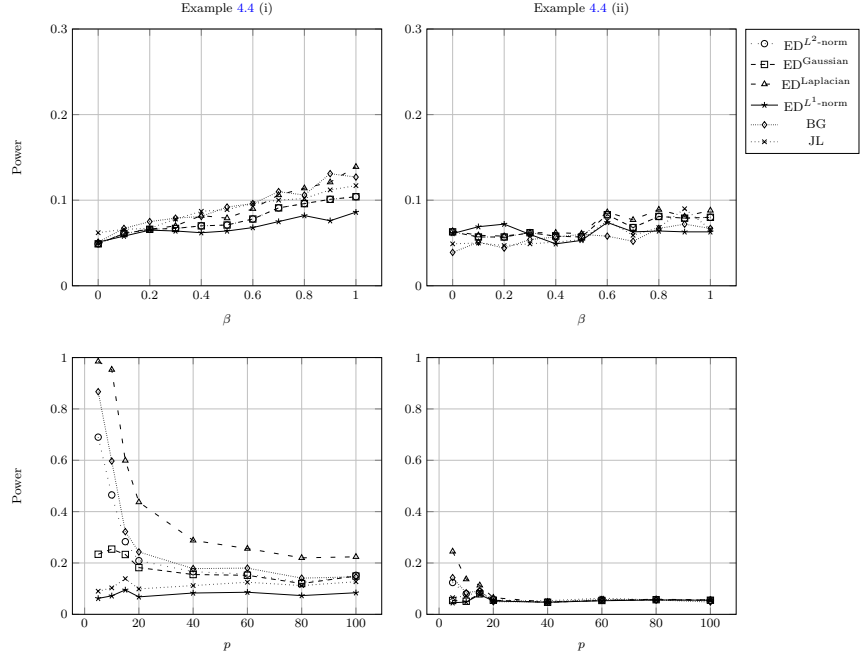
	
	\subsection{Performance on real data}
	\label{sub:1}
	We also compare the power of the above tests on the following real data sets.
	
	\begin{figure}
		\centering
		\begin{tikzpicture}[scale = 0.81]
		\begin{axis}[
		name=straw,
		ylabel=Strawberry,
		width=\textwidth,
		height=0.3\textwidth,
		ytick=\empty,
		xtick=\empty,
		every axis plot/.append style={ultra thick},
		legend pos=outer north east]
		\addplot[blue, mark=none] table[x index = 0, y index = 1, col sep=comma] {strawberry.dat};
		\addlegendentry{Class 1}
		\addplot[red, mark=none, dashed] table[x index =0, y index = 2, col sep=comma] {strawberry.dat};
		\addlegendentry{Class 2}
		\end{axis}
		
		\begin{axis}[
		name = kitchen,
		at ={(straw.below south west)},
		yshift=-3.5cm,
		ylabel=KitchenAppliances,
		width=\textwidth,
		height=0.3\textwidth,
		ytick=\empty,
		xtick=\empty,
		every axis plot/.append style={ultra thick}
		]
		\addplot[blue, mark=none] table[x index = 0, y index = 1, col sep=comma] {SmallKitchenAppliances.dat};
		\addlegendentry{Class 1}
		\addplot[red, mark=none, dashed] table[x index =0, y index = 2, col sep=comma] {SmallKitchenAppliances.dat};
		\addlegendentry{Class 2}
		\legend{}
		\end{axis}
		
		\begin{axis}[
		name = worm,
		at ={(kitchen.below south west)},
		yshift=-3.5cm,
		ylabel=Earthquakes,
		width=\textwidth,
		height=0.3\textwidth,
		ytick=\empty,
		xtick=\empty,
		every axis plot/.append style={ultra thick}
		]
		\addplot[blue, mark=none] table[x index = 0, y index = 1, col sep=comma] {Earthquakes.dat};
		\addlegendentry{Class 1}
		\addplot[red, mark=none, dashed] table[x index =0, y index = 2, col sep=comma] {Earthquakes.dat};
		\addlegendentry{Class 2}
		\legend{}
		\end{axis}
		\end{tikzpicture}
		\caption{A glance of the data in Section \ref{sub:1}, where we plot one point from each of the two classes for each data set.}
		\label{figure4}
	\end{figure}
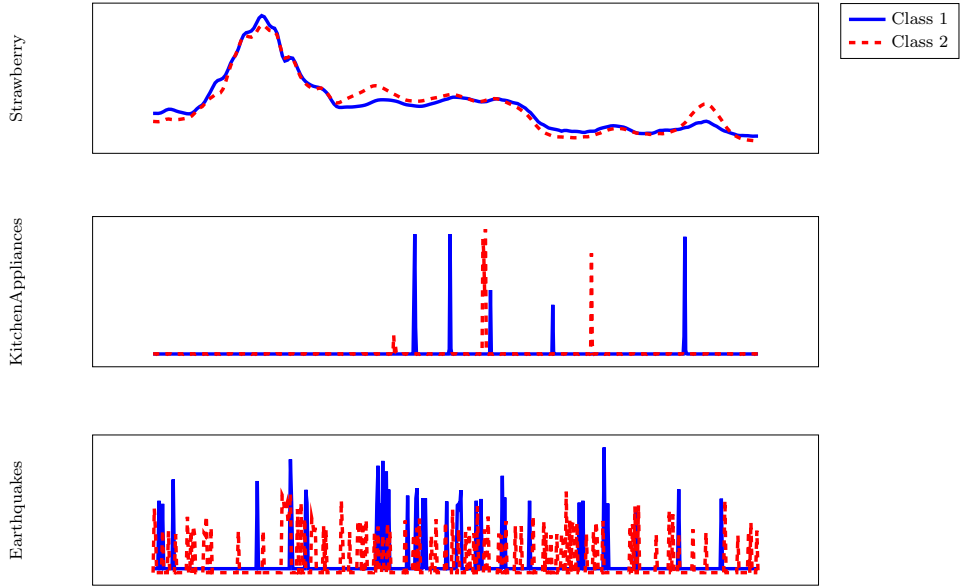
	
	\begin{itemize}
		\item Strawberry data: this data set contains the spectrographs of fruit purees. There are totally two classes: one is strawberry purees (authentic samples) and the other one is non-strawberry purees (adulterated strawberries and other fruits). Each data point is of length 235.
		\item SmallKitchenAppliances data: this data sets contains records of the electricity usage of some kitchen appliances. We only use classes Kettle and Microwave. Each data point has readings taken every 2 minutes over 24 hours.
		\item Earthquakes data: this data set is from Northern California Earthquake Data Center and has classes of positive and negative major earthquake events. There are 368 negative and 93 positive cases and each data point is of length 512.
	\end{itemize}
	
	All the above data sets are downloaded from UCR Time Series
	Classification Archive \cite{UCRArchive2018} (\url{https://www.cs.ucr.edu/~eamonn/time_series_data_2018/}) and a glance of these data sets is provided in Figure \ref{figure4}. For each of the three data sets, the data points have two classes and we want to compare the underlining distributions of the two classes. Following the procedures of \cite{biswas2014nonparametric} and \cite{sarkar2018high}, for each $m=n \in \{ 10,20,30,40,50,60 \}$, we randomly sample $n$ points from each class and test  whether the two distributions are the same using the afore-mentioned tests. The same procedure is repeated 1000 times to calculate the power.
	
	The experimental results for these data sets are shown in Figure \ref{figure5}, from which we see that all the tests have very high power for the Strawberry data with relatively low sample size. As for the SmallKitchenAppliances and Earthquakes data sets, the $L^1$-norm based test demonstrates superior power compared to other tests. It is also worth noting that BG and JL  barely exhibit any power for the Earthquakes data.
	\begin{figure}[ht]
		\centering
		\begin{tikzpicture}[scale = 0.5]
		\begin{axis}[
		name = ax1,
		title= {Strawberry} ,
		xlabel={$ n=m $},
		ylabel={Power},
		ymin=0, ymax=1,
		grid = major,
		]
		\addlegendentry{$\text{ED}^{L^2\text{-norm}}$}
		\addplot [mark=o,loosely dotted, every mark/.append style={solid}] table[x index=0, y index=1, col sep=comma] {r11.dat};
		
		\addlegendentry{$\text{ED}^{text{Gaussian}}$}
		\addplot [mark=square, dashed, every mark/.append style={solid}] table[x index=0, y index=2, col sep=comma] {r11.dat};
		
		\addlegendentry{$\text{ED}^{text{Laplacian}}$}
		\addplot [mark=triangle, loosely dashed, every mark/.append style={solid}] table[x index=0, y index=3, col sep=comma] {r11.dat};
		
		\addlegendentry{$\text{ED}^{L^{1}}$}
		\addplot table[x index=0, y index=4, col sep=comma] {r11.dat};
		
		\addlegendentry{$\text{BG}$}
		\addplot [mark=diamond, densely dotted, every mark/.append style={solid}]  table[x index=0, y index=5, col sep=comma] {r11.dat};
		
		\addlegendentry{$\text{JL}$}
		\addplot [mark=x, dotted, every mark/.append style={solid}] table[x index=0, y index=6, col sep=comma] {r11.dat};
		
		\legend{}
		\end{axis}
		
		\begin{axis}[
		name = bx1,
		at={(ax1.south east)},
		xshift=1cm,
		title= {SmallKitchenAppliances} ,
		xlabel={$ n=m $},
		ymin=0, ymax=1,
		grid = major,
		legend pos=outer north east,
		]
		\addlegendentry{$\text{ED}^{L^2\text{-norm}}$}
		\addplot [mark=o,loosely dotted, every mark/.append style={solid}] table[x index=0, y index=1, col sep=comma] {r12.dat};
		
		\addlegendentry{$\text{ED}^{\text{Gaussian}}$}
		\addplot [mark=square, dashed, every mark/.append style={solid}] table[x index=0, y index=2, col sep=comma] {r12.dat};
		
		\addlegendentry{$\text{ED}^{\text{Laplacian}}$}
		\addplot [mark=triangle, loosely dashed, every mark/.append style={solid}] table[x index=0, y index=3, col sep=comma] {r12.dat};
		
		\addlegendentry{$\text{ED}^{L^{1}\text{-norm}}$}
		\addplot table[x index=0, y index=4, col sep=comma] {r12.dat};
		
		\addlegendentry{$\text{BG}$}
		\addplot [mark=diamond, densely dotted, every mark/.append style={solid}]  table[x index=0, y index=5, col sep=comma] {r12.dat};
		
		\addlegendentry{$\text{JL}$}
		\addplot [mark=x, dotted, every mark/.append style={solid}] table[x index=0, y index=6, col sep=comma] {r12.dat};
		\legend{}
		\end{axis}

		\begin{axis}[
		name = x1,
		at={(bx1.south east)},
		xshift=1cm,
		title= {Earthquakes} ,
		xlabel={$ n=m $},
		ymin=0, ymax=1,
		grid = major,
		legend style={at={(-0.6,-0.2)},anchor=north, legend columns=-1}
		]
		\addlegendentry{$\text{ED}^{L^2\text{-norm}}$}
		\addplot [mark=o,loosely dotted, every mark/.append style={solid}] table[x index=0, y index=1, col sep=comma] {r21.dat};
		
		\addlegendentry{$\text{ED}^{\text{Gaussian}}$}
		\addplot [mark=square, dashed, every mark/.append style={solid}] table[x index=0, y index=2, col sep=comma] {r21.dat};
		
		\addlegendentry{$\text{ED}^{\text{Laplacian}}$}
		\addplot [mark=triangle, loosely dashed, every mark/.append style={solid}] table[x index=0, y index=3, col sep=comma] {r21.dat};
		
		\addlegendentry{$\text{ED}^{L^{1}}$}
		\addplot table[x index=0, y index=4, col sep=comma] {r21.dat};
		
		\addlegendentry{$\text{BG}$}
		\addplot [mark=diamond, densely dotted, every mark/.append style={solid}]  table[x index=0, y index=5, col sep=comma] {r21.dat};
		
		\addlegendentry{$\text{JL}$}
		\addplot [mark=x, dotted, every mark/.append style={solid}] table[x index=0, y index=6, col sep=comma] {r21.dat};
		
		\end{axis}
		
		\end{tikzpicture}
		\caption{Power comparison for real data examples in Section \ref{sub:1}.}
		\label{figure5}
	\end{figure}
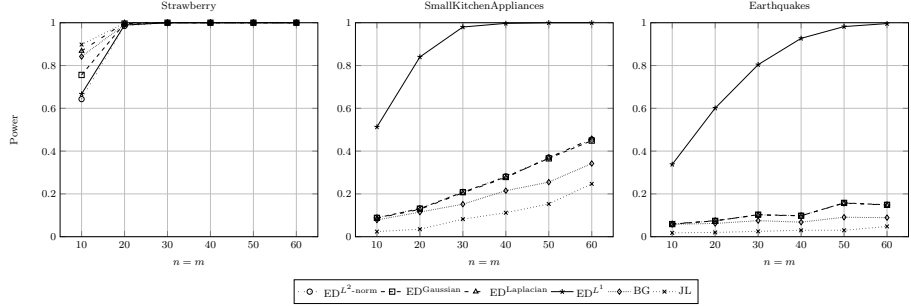
	
	\section{Discussions\&Conclusion}
	\label{sec:twosample:con}
	

	In this paper, we study the two-sample hypothesis testing problem in a high dimension and low/medium
	sample size setting. Our focus is on the interpoint distance based permutation tests, such as those based on  Energy
	Distance (ED) and Maximum Mean Discrepancy (MMD). Our theory demonstrates that all these tests under examination
	are unable to detect the difference between two high dimensional distributions beyond univariate marginal distributions.
	In particular, the ED test with $L^2$-norm and MMD with Gaussian or Laplacian kernels suffer substantial power loss under the HDLSS and have trivial power under the HDMSS
	when   the average of component-wise mean and variance discrepancies between two distributions are both asymptotically
	zero at the rate of $o(1/\sqrt{nmp})$. Thus these tests mainly target mean and variance differences in marginal distributions. By contrast, if we
	use $L^1$-norm in ED test, then the non-negligible difference in marginal univariate distributions, as quantified by cumulative
	energy distance of marginal distributions, can be detected with high power. Thus the theory suggests that
	
	1), The ED with $L^2$-norm, and MMD with Gaussian and Laplacian kernels are of the same category, as they all depend on
	the interpoint distance as measured by Euclidean distance, which leads to undesirable power limitation.
	
	2), Although in a low dimensional setting the use of $L^1$-norm in ED is not preferred due to the fact that it does not completely
	characterize the difference between two distributions since $\text{ED}^1(F,G)=0$ does not necessarily imply $F=G$, it seems to have some advantage over the ED with $L^2$-norm and MMD with Gaussian and
	Laplacian kernels  in the high dimensional setting, as shown in both theory and numerical studies.
	
	3), As shown in our simulations and data illustration, the existing interpoint distance test by \cite{li2018asymptotic} and \cite{biswas2014nonparametric} also suffer from low power when the two distributions have the same marginal mean and variances but different marginal distributions. So in this sense, they are also
	inferior to the ED test with $L^1$-norm.
	
	4), The difference in marginal distributions of two high dimensional distributions can be interpreted as the main effect of the distribution differences. It is
	a standard statistical practice to  test for the nullity of main effects first, before proceeding to the higher-order interactions. Thus we advocate the use of $L^1$-norm based test to test for the presence of main differences in two high dimensional distributions.
	
	To conclude the paper, we shall mention a few future directions. First, we are holding the bandwidth parameter in
	Gaussian and Laplacian kernels fixed for theoretical convenience, and it would be interesting to relax this restriction by allowing it to be data-dependent. Second, there might be some intrinsic difficulty of capturing all kinds of differences in two high dimensional distributions with limited sample sizes, so it seems natural to ask whether it is possible to detect any difference beyond marginal univariate distributions. If possible, what would be the form of the new tests?
	We leave these topics for future investigation.

	\appendix
	\section{Technical Details}
	\subsection{Proof of Sufficient Conditions for Local Alternatives} \label{subsec:hac}
	When $\psi(x,y) = (x-y)^2$, $\varphi$ is strictly concave, strictly increasing on $(0, +\infty)$ (e.g. $L^2$-norm, Gaussian kernel myltiplied by -1 and Laplacian kernel multiplied by -1), we first note that $2 e_{xy} - e_x - e_y =  2 \lim_{p \rightarrow \infty} \sum_{u=1}^{p} (E(x_u)-E(y_u))^2/p \geq 0$  and 
	\begin{align*}
	\varphi(e_{xy}) -\frac{ \varphi(e_x) + \varphi(e_y)}{2} \geq  \varphi(e_{xy}) - \varphi \left( \frac{e_{x} + e_y}{2} \right) \geq 0,
	\end{align*}
	where the equality holds iff $e_{xy} = e_{x} = e_{y}$. Also, some algebra shows that
	\begin{multline*}
	e_{xy} = e_x  + \lim\limits_{p \rightarrow \infty} \frac{1}{p} \sum_{u=1}^{p} (E(x_u)-E(y_u))^2 + \lim\limits_{p \rightarrow \infty} \frac{1}{p} \sum_{u=1}^{p} (var(y_u)-var(x_u) )  \\= e_y  + \lim\limits_{p \rightarrow \infty} \frac{1}{p} \sum_{u=1}^{p} (E(x_u)-E(y_u))^2 + \lim\limits_{p \rightarrow \infty} \frac{1}{p} \sum_{u=1}^{p} (var(x_u)-var(y_u) ).
	\end{multline*}
	Thus, in summary we have
	\begin{multline*}
	2 \varphi(e_{xy}) = \varphi(e_x) + \varphi(e_y) \Leftrightarrow e_{xy} = e_{x} = e_{y} \\ \Leftrightarrow 
	\sum_{u=1}^{p} (E(x_u)-E(y_u))^2 = o(p)  \text{ and } \left|  \sum_{u=1}^{p} (var(x_u) - var(y_u)) \right|= o(p).
	\end{multline*}
	This proves the result for $H_{A_c}$ characterization. Next, for sufficient conditions of $H_{A_l}$, if we have 
	$$
	\begin{array}{l}
	\sum\limits_{u=1}^{p} (E(x_u)-E(y_u))^2 = o\left(\frac{\sqrt{p}}{\sqrt{nm}}\right) \text{ and }
	\left| \sum\limits_{u=1}^{p} (var(x_u) - var(y_u)) \right| =  o\left(\frac{\sqrt{p}}{\sqrt{nm}}\right),
	\end{array}
	$$ 
	then it holds that $e_{xy} = e_x=e_y$ and 
	\begin{multline*}
	 E \left[ \left|  E\left[\overline{\psi}(X,Y)|X \right] -  E\left [\overline{\psi}(X,X')|X \right] \right| \right] \\
	\leq  \frac{2}{p} \sqrt{ \sum_{u=1}^p  E(x_u^2)  \sum_{u=1}^p ( E(x_u) - E(y_u))^2 } + \frac{1}{p} \left|  \sum_{u=1}^p (\var(y_u) - \var(x_u)) \right| \\
	 \hspace{1cm}  + \frac{1}{p}  \sqrt{ \sum_{u=1}^p  (E(x_u) + E(y_u))^2  \sum_{u=1}^p ( E(x_u) - E(y_u))^2 }= o \left( \frac{1}{\sqrt{nmp}}\right).
	\end{multline*} 
	For $H_{A_t}$, a straight forward calculation shows that
	\begin{align*}
	\psi (x, y) - E\left[\psi (x, y)|x\right]  -  E\left[\psi (x, y)|y\right] +  E\left[\psi (x , y)\right]  = -2(x-E(x))(y-E(y))
	\end{align*}
	and
	$
	v_{xy} = \sum_{u,v=1}^p 4\cov(x_u, x_v) \cov (y_u, y_v)/p.
	$ Thus, from Cauchy–Schwarz inequality, we have 
	$$
	\sum_{u,v=1}^{p} (\cov(x_u, x_v) - \cov (y_u, y_v) )^2 = o(p) \Rightarrow v_{xy} = v_x = v_y.
	$$
	When $ \psi(x,y) = |x-y| $, $\varphi (x)=x$, the results follow from the following equality.
	\begin{align*}
	& 2 \varphi(e_{xy}) - \varphi(e_x) - \varphi (e_y) \\
	= & 2 \lim\limits_{p \rightarrow \infty} \frac{1}{p} \sum_{u=1}^{p} E \left[ |x_{1u} - y_{1u}| \right] -  \lim\limits_{p \rightarrow \infty} \frac{1}{p} \sum_{u=1}^{p} E \left[ |x_{1u} - x_{2u}| \right] - \lim\limits_{p \rightarrow \infty} \frac{1}{p} \sum_{u=1}^{p} E \left[ |y_{1u} - y_{2u}| \right] \\
	= & \lim\limits_{p \rightarrow \infty} \frac{1}{p} \sum_{u=1}^{p}\left( 2E \left[ |x_{1u} - y_{1u}| \right] - E \left[ |x_{1u} - x_{2u}| \right] -  E \left[ |y_{1u} - y_{2u}| \right] \right) \\
	= & \lim\limits_{p \rightarrow \infty} \frac{1}{p} \sum_{u=1}^{p} \text{ED}(F_u, G_u).
	\end{align*}
	
	\subsection{Proof of Theorem \ref{thm:lowall}}
	\begin{proof}
		(i)  Taking a first order Taylor expansion w.r.t $\varphi$ gives
		\begin{align*}
		k(Z_i, Z_j)  = \varphi \left( e_{ij} \right) + \mathcal{R}_{1}(Z_i, Z_j),
		\end{align*}
		where  $\mathcal{R}_{1}(Z_i, Z_j) $ is an operator that acts on random variables 
		\begin{align*}
		\mathcal{R}_{1}(Z_i, Z_j) =  \mathcal{L}(Z_i, Z_{j})\int_{0}^{1} \varphi^{(1)} \left(e_{ij} + v \mathcal{L}(Z_i, Z_{j}) \right) dv
		\end{align*}
		For each fixed permutation matrix $ \varGamma_w \in \mathbb{S}_{w}$
		\begin{align} \label{eq:dep1}
		\text{ED}_n^{k}( \varGamma_w \mathbf{Z}) & = \underbrace{\sum_{i=2}^{n+m}\sum_{j=1}^{i-1} \varPi_{w,ij} \varphi \left( e_{ij} \right)}_{:= \mu_n ( \varGamma_w \mathbf{ Z})} + \underbrace{\sum_{i=2}^{n+m}\sum_{j=1}^{i-1}  \varPi_{w,ij} \mathcal{R}_{1}(Z_i, Z_j)}_{:= R_1(\varGamma_w \mathbf{ Z})},
		\end{align}
		where $\mu_n ( \varGamma_w \mathbf{ Z})$ is the asymptotic mean for the permuted data and equals
		\begin{multline*}
		\frac{2}{mn} \left( \left(w^{2} + (n-w)(m-w)\right) \varphi (e_{xy} ) + (n-w)w \varphi( e_{x}) + (m-w)w \varphi(e_{y} ) \right) \\
		-  \frac{1}{n(n-1)} \left(2w(n-w) \varphi( e_{xy} ) + (n-w)(n-w-1)\varphi(e_{x}) + w(w-1) \varphi(e_{y}) \right) \\
		-  \frac{1}{m(m-1)} \left(2w(m-w) \varphi (e_{xy})    + w(w-1) \varphi( e_{x}) + (m-w)(m-w-1) \varphi( e_{y} )\right).
		\end{multline*}
		Then, after re-arranging the terms according to the powers of $w$, we have 
		$$
		\mu_n  (\varGamma_w \mathbf{Z}) =  \left( 2\varphi(e_{xy}) - \varphi(e_{x}) - \varphi(e_{y}) \right) f(w),
		$$
		where $f(w)$ is a second order polynomial with respect to $w$
		\begin{align*}
		f(w) =  1
		-  \left( \frac{2m-1}{m(m-1)} + \frac{2n-1}{n(n-1)} \right)w
		+  \left( \frac{2}{mn}+ \frac{1}{n(n-1)} + \frac{1}{m(m-1)} \right) w^{2} .
		\end{align*}
		For the remainder term $R_1( \varGamma_w \mathbf{ Z})$,
		notice that $\mathcal{L}(Z_i, Z_j) \overset{p}{\rightarrow} 0$ for any $1 \leq i < j < n+m$. By the continuous mapping theorem, we know
		\begin{align*}
		\int_{0}^{1} \varphi^{(1)}(e_{ij} + v\mathcal{L}(Z_i, Z_j)) dv \overset{p}{\rightarrow}  \int_{0}^{1} \varphi^{(1)}(e_{ij}) dv.
		\end{align*}
		Thus, it holds that $ \mathcal{R}_1(Z_i, Z_j) \asymp_p \mathcal{L}(Z_i, Z_j) $ and $R_{1}( \varGamma_w \mathbf Z) = O_p(\alpha_{xy}+ \alpha_{x} + \alpha_{y}) = o_p(1)$.
		
	(ii)  Taking a second order Taylor expansion w.r.t $\varphi$ gives
	\begin{align*}
	k(Z_i, Z_j)  = \varphi \left(e_{ij} \right) +  \varphi^{(1)} \left(e_{ij} \right) \frac{ \mathcal{K}(Z_i, Z_j) + \mathcal{W}(Z_i, Z_j) }{\sqrt{p}}   + \mathcal{R}_{2}(Z_i, Z_j),
	\end{align*}
	where $\mathcal{R}_{2} $ and $\mathcal{W}$ are defined as
	\begin{align*}
	\mathcal{R}_{2}(Z_i, Z_j) & = \mathcal{L}^2(Z_i, Z_{j}) \int_{0}^{1} \int_{0}^{1} u \varphi^{(2)}\left(e_{ij} + uv \mathcal{L}(Z_i, Z_{j}) \right) dvdu, \\
	\mathcal{W}(Z_i, Z_j) & = \frac{1}{\sqrt{p}} \sum_{u=1}^{p} \left(E[ \psi (z_{iu}, z_{ju})| z_{iu}] +  E[ \psi (z_{iu}, z_{ju})| z_{ju}] - E[\psi(z_{iu}, z_{ju})] - e_{ij} \right) .
	\end{align*}
	Accordingly, we can decompose the sample energy distance as 
	\begin{multline} \label{eq:dep2}
	\sqrt{p} \left( \text{ED}_n^{k}(\varGamma_w \mathbf{Z}) - \mu_n ( \varGamma_w \mathbf Z)  \right)    =   \underbrace{ \sum_{i=2}^{n+m}\sum_{j=1}^{i-1}  \varPi_{w,ij}\varphi^{(1)}\left( e_{ij} \right) \mathcal{K}(Z_i, Z_j)}_{:=L( \varGamma_w \mathbf{ Z} )}  \\ + \underbrace{ \sum_{i=2}^{n+m}\sum_{j=1}^{i-1}  \varPi_{w,ij} \varphi^{(1)}(e_{ij}) \mathcal{W}(Z_i, Z_j)}_{:=R_l( \varGamma_w \mathbf{ Z})} + \underbrace{ \sqrt{p}\sum_{i=2}^{n+m}\sum_{j=1}^{i-1} \varPi_{w,ij}\mathcal{R}_2(Z_i,Z_j)}_{:= R_{2}( \varGamma_w \mathbf{ Z}) }.
	\end{multline}
	For the leading term $L(\varGamma_w \mathbf Z)$, notice that under Assumption \ref{Ass:1},  $( \mathcal{K}(Z_i, Z_j))_{i<j} $ converges jointly to a multivariate normal with mean 0 and a diagonal covariance matrix. Thus, given a permutation matrix $\varGamma_w$, we are able to obtain $ L(\varGamma_w \mathbf Z) \overset{d}{\rightarrow} N(0, \sigma^2_n(\varGamma_w \mathbf Z)) $, where
		\begin{align*}
		\sigma^2_n(\varGamma_w \mathbf Z) =  & \frac{4}{n^2(n-1)^2} \bigg\{ \frac{(n -w)(n -w-1)}{2}v_x[\varphi^{(1)}(e_{x})]^2 \\
		& \hspace{1cm} + \frac{w (w -1) }{2} v_y[\varphi^{(1)}(e_{y})]^2 + (n -w)w v_{xy} [\varphi^{(1)}(e_{xy})]^2 \bigg\} \\
		+ & \frac{4}{m^2(m-1)^2}\bigg\{ \frac{w (w -1)}{2}v_x[ \varphi^{(1)}(e_{x})]^2 \\
		& \hspace{1cm} + \frac{(m-w)(m-w-1)}{2}v_y[ \varphi^{(1)}(e_{y})]^2 + w(m-w)v_{xy}[ \varphi^{(1)}(e_{xy})]^2 \bigg\} \\ 
		+ & \frac{4}{n^2m^2} \bigg\{ (n-w)w v_x[ \varphi^{(1)}(e_{x})]^2 \\
		& \hspace{1cm} + w(m-w)v_y[ \varphi^{(1)}(e_{y})]^2 + ((n-w)(m-w)+w^2)v_{xy} [ \varphi^{(1)}(e_{xy})]^2  \bigg\}.
	\end{align*}
	By collecting terms with respect to $v_{xy}, v_x, v_y$, we obtain
	\begin{align*}
	\sigma^2_{n,w}= &  \bigg\{ \frac{4}{nm} -4 \left( \frac{n+m}{n^2m^2} - \frac{n}{n^2(n-1)^2} - \frac{m}{m^2(m-1)^2} \right)w \\ 
		& \hspace{1cm}+ 4 \left(\frac{2}{n^2m^2} - \frac{1}{n^2(n-1)^2} - \frac{1}{m^2(m-1)^2} \right) w^{2} \bigg\} v_{xy} [\varphi^{(1)}(e_{xy})]^2  \\
		+ & \bigg\{ \frac{2}{n(n-1)} + 2\left( \frac{2n}{n^2m^2} - \frac{2n-1}{n^2(n-1)^2} - \frac{1}{m^2(m-1)^2} \right) w \\
		& \hspace{1cm} - 2\left(\frac{2}{m^2n^2} - \frac{1}{n^2(n-1)^2} - \frac{1}{m^2(m-1)^2} \right) w^{2}  \bigg\} v_x[ \varphi^{(1)}(e_{x})]^2 \\
		+ & \bigg\{ \frac{2}{m(m-1)} + 2\left( \frac{2m}{n^2m^2}  - \frac{1}{n^2(n-1)^2} - \frac{2m-1}{m^2(m-1)^2} \right) w \\
		& \hspace{1cm} -2 \left(\frac{2}{n^2m^2} - \frac{1}{m^2(m-1)^2} - \frac{1}{n^2(n-1)^2} \right) w^{2}  \bigg\} v_y[ \varphi^{(1)}(e_{y})]^2.
		\end{align*}
	We then conclude the result by showing that the remainder terms are negligible. $R_l( \varGamma_w \mathbf{ Z}) = o_p(1)$ is proved in lemma \ref{lemma:lowRl}. For the $R_2( \varGamma_w \mathbf Z)$ term, it can be shown similarly that $\mathcal{R}_2(Z_i, Z_j) \asymp_p \mathcal{L}^2(Z_i, Z_j) = O_p(\alpha_{xy}^2 + \alpha_{x}^2 + \alpha_y^2) $, which implies that $R_2( \varGamma_w \mathbf Z) = O_p(\sqrt{p}(\alpha_{xy}^2 + \alpha_{x}^2 + \alpha_y^2) ) = o_p(1)$ under Assumption \ref{Ass:4}. 
	\end{proof}

	\subsection{Proof of Proposition \ref{Lemmalow: asymHp}}
\begin{proof}
	(i) From Equation \ref{eq:dep1}, we obtain
	\begin{align*}
	\text{ED}_n^{k}(\mathbf \Gamma \mathbf{Z}) & = \underbrace{\sum_{i=2}^{n+m}\sum_{j=1}^{i-1} \mathbf \Pi_{ij} \varphi \left( e_{ij} \right)}_{:= \mu_{n, W}} +o_p(1),
	\end{align*}
	where $\mathbf{\Pi}_{ij}$ corresponds to $\mathbf{\Gamma}$ and $W = N(\mathbf \Gamma) \sim \text{Hypergeometric}(m+n,m,n)$.
	
	(ii) It follows from Equation \ref{eq:dep2}, Lemma \ref{lemma:lowRl} and the proof of Theorem \ref{thm:lowall} that
	\begin{align*}
	\sqrt{p} \left( \text{ED}_n^{k}(\mathbf \Gamma \mathbf{Z}) - \mu_n ( \mathbf \Gamma \mathbf Z)  \right)    =   \underbrace{ \sum_{i=2}^{n+m}\sum_{j=1}^{i-1} \mathbf \Pi_{ij}\varphi^{(1)}\left( e_{ij} \right) \mathcal{K}(Z_i, Z_j)}_{:=L(\mathbf \Gamma \mathbf{ Z} )} + o_p(1).
	\end{align*} 
	Then, under Assumption \ref{Ass:1}, it is not hard to see that $ L(\mathbf \Gamma \mathbf{ Z} ) \overset{d}{\rightarrow} N(0, \sigma^2_{n,W}) $, where $W = N(\mathbf \Gamma) \sim \text{Hypergeometric}(m+n,m,n)$. This concludes the proposition.
\end{proof}
	
	\subsection{Proof of Theorem \ref{lemma:power1}}
	1, For any $\mathbf{a} \in \mathbb{R}^{(n+m)!}$, we define the $\alpha$-th quantile of the set $ \{ a_{1}, \cdots , a_{(n+m)!} \} $ as
	$$
	Q_{1-\alpha}\left\lbrace  a_{1}, \cdots , a_{(n+m)!} \right\rbrace  = \min \left\lbrace  a_{i}:  \frac{1}{(n+m)!} \sum\limits_{i=1}^{(n+m)!} \mathbb{I}_{\{ a_{i} \leq t \} } \geq 1 - \alpha \right\rbrace.
	$$
	Then, we can view $Q_{1 - \alpha}$ as a continuous function on $\mathbb{R}^{(n+m)!}$.
	
	(i) By Theorem \ref{thm:lowall}, for any fixed $\Gamma_{i} \in \mathbb{P}_{n+m}$, we have
	$
	\text{ED}_n^k( \Gamma_i \mathbf{Z})   \overset{p}{\rightarrow}  \mu_n( \Gamma_i \mathbf{Z}).
	$
	The continuous mapping theorem implies
	$$
	Q_{1-\alpha} \left\{ \text{ED}_n^k(\Gamma_{1} \mathbf{Z}),
	\cdots,
	\text{ED}_n^k(\Gamma_{(n+m)!} \mathbf{Z}) \right\} 
	\overset{p}{\rightarrow} Q_{1 - \alpha}
	\left\{ 
	\mu_n( \Gamma_{1} \mathbf{Z}),
	\cdots ,
	\mu_n( \Gamma_{(n+m)!} \mathbf{Z}) \right\}.
	$$
	Then, it follows from the definition of $\mu_n(\cdot)$ in Theorem \ref{thm:lowall} that
	$$
	\mu_n(\mathbf Z) = \mu_n (\Gamma_{0} \mathbf{Z}) = \max \left\lbrace \mu_n( \Gamma_{1} \mathbf Z ), \cdots , \mu_n( \Gamma_{(n+m)!} \mathbf Z)  \right\rbrace .
	$$
	Notice that 
	$$\left\lbrace
	\begin{array}{ll}
	\frac{n!m!}{(n+m)!} < 1 - \alpha & \text{if }m \neq n, \\
	\frac{2 (n!)^2}{(2n)!} < 1 - \alpha & \text{if }m = n,
	\end{array}\right.
	\text{ implies } \left\lbrace
	\begin{array}{ll}
	\frac{|\mathbb{S}_0| }{(n+m)!} < 1 - \alpha & \text{if }m \neq n, \\
	\frac{|\mathbb{S}_0| + |\mathbb{S}_{\min \{n,m\}}| }{(n+m)!} < 1 - \alpha & \text{if }m = n,
	\end{array}\right.
	$$
	and so $\mu_n (  \Gamma_{0} \mathbf{Z}) > Q_{1- \alpha} \left\lbrace \mu_n( \Gamma_{1} \mathbf Z ), \cdots , \mu_n( \Gamma_{(n+m)!} \mathbf Z)  \right\rbrace$.
	Thus, as $p \rightarrow \infty$, we conclude
	\begin{multline*}
	P \left( \text{ED}_n^k (\mathbf{Z}) > Q_{1 - \alpha} \left\lbrace \text{ED}_n^k(\Gamma_{1} \mathbf{Z}) , \cdots ,  \text{ED}_n^k(\Gamma_{(n+m)! } \mathbf{Z})  \right\rbrace \right) \\
	\rightarrow   P\left(\mu_n(\Gamma_{0} \mathbf{Z}) > Q_{1 - \alpha} \left\lbrace \mu_n( \Gamma_{1} \mathbf Z ), \cdots , \mu_n( \Gamma_{(n+m)!} \mathbf Z)  \right\rbrace \right) =  1.
	\end{multline*}
	(ii) For any random permutation matrix $\mathbf{\Gamma}_s$, by Proposition \ref{Lemmalow: asymHp}, we have $\text{ED}_n^k( \mathbf \Gamma_s \mathbf{Z})   \overset{p}{\rightarrow}  \mu_{n, W_s}$, where $W_s  = N(\mathbf{\Gamma}_s) \sim \text{Hypergeometric}(n+m,m,n)$. Then, the continuous mapping theorem implies that
	\begin{multline*}
	P \left( \text{ED}_n^k (\mathbf{Z}) > Q_{1 - \alpha} \left\lbrace \text{ED}_n^k( \mathbf{Z}) , \text{ED}_n^k(\mathbf{\Gamma}_1 \mathbf{Z}), \cdots ,  \text{ED}_n^k(\mathbf{\Gamma}_{S-1} \mathbf{Z})  \right\rbrace \right) \\
	\rightarrow   P\left(\mu_{n,0} > Q_{1 - \alpha} \left\lbrace \mu_{n,0}, \mu_{n,W_1}, \cdots , \mu_{n,W_{S-1}}  \right\rbrace \right).
	\end{multline*}
	Since $\mu_{n,0} = \max_{w} \mu_{n,w}$, in order to have $\mu_{n,0} = Q_{1 - \alpha} \left\lbrace \mu_{n,0}, \mu_{n,W_1}, \cdots , \mu_{n,W_{S-1}}  \right\rbrace$, at least $ \lfloor \alpha S \rfloor +1 $ elements of $\left\lbrace \mu_{n,0}, \mu_{n,W_1}, \cdots , \mu_{n,W_{S-1}}  \right\rbrace$ should be equal to $\mu_{n,0}$. Thus, we get
	\begin{multline*}
	 P\left(\mu_{n,0} > Q_{1 - \alpha} \left\lbrace \mu_{n,0} , \mu_{n,W_1}, \cdots , \mu_{n,W_{S-1}}  \right\rbrace \right) \\ = 1 - 
	 P\left(\mu_{n,0} = Q_{1 - \alpha} \left\lbrace \mu_{n,0} , \mu_{n,W_1}, \cdots , \mu_{n,W_{S-1}}  \right\rbrace \right)
	 \\ \geq  \left\{ 
	 \begin{array}{ll}
	 1 - \frac{S-1}{ \lfloor \alpha S \rfloor } \frac{n!m!}{(n+m)!}, & \text{if } n \neq m , \\
	 1 - \frac{S-1}{ \lfloor \alpha S \rfloor } \frac{2(n!)^2}{(n+m)!}, & \text{if } n = m .
	 \end{array}	\right.
	\end{multline*}

	2, (i) Since $ \mu_n ( \Gamma_u \mathbf Z) =0 $ for all $u=1,2, \cdots, (n+m)!$ under $H_{A_l}$, Assumption \ref{Ass:1} implies that
	\begin{align*}
	\sqrt{p}  \text{ED}_n^{k}(\Gamma_u \mathbf{Z})    \overset{d}{\rightarrow   }  \sum\limits_{i=1}^{n} \sum\limits_{j=1}^{m} \Pi_{u,ij} b_{ij}  -   \sum\limits_{1 \leq i < j \leq n} \Pi_{u,ij} c_{ij} -  \sum\limits_{1 \leq i < j \leq m} \Pi_{u,ij} d_{ij},
	\end{align*}
	where $ \Pi_{u,ij} $ corresponds to $\Gamma_u$. Then, the continuous mapping theorem entails
	\begin{multline*}
		 P \left( \text{ED}_n^k (\mathbf{Z}) > Q_{1 - \alpha} \left\lbrace \text{ED}_n^k(\Gamma_{1} \mathbf{Z}) , \cdots ,  \text{ED}_n^k(\Gamma_{(n+m)! } \mathbf{Z})  \right\rbrace \right) \\
		=   P \left( \sqrt{p} \text{ED}_n^k (\mathbf{Z}) >  Q_{1 - \alpha} \left\lbrace \sqrt{p} \text{ED}_n^k(\Gamma_{1} \mathbf{Z}) , \cdots , \sqrt{p} \text{ED}_n^k(\Gamma_{(n+m)! } \mathbf{Z})  \right\rbrace \right)\\
		\rightarrow   	P(V(\Gamma_0) >  Q_{\widehat{T}, 1-\alpha} ) .
	\end{multline*}
	
	(ii) Conditioned on $\mathbf{\Gamma}_1,\mathbf{\Gamma}_2, \cdots, \mathbf{\Gamma}_{S-1}$, the result can be shown similarly with part(i). Then, since the number of permutations is fixed and finite, the unconditioned version follows straightforwardly.
	
	3, (i)  By construction, we have
	\begin{align*}
 	\frac{1}{(n+m)!} \sum\limits_{u=1}^{(n+m)!} \mathbb{I}_{ \left\lbrace V(\Gamma_u)  > Q_{1-\alpha} \left\{  V(\Gamma_1), \cdots, V(\Gamma_{(n+m)!}) \right\} \right\rbrace }  \leq \alpha.
	\end{align*}
	If $\varphi' ( e_{xy}) = \varphi'( e_{x}) = \varphi'( e_{y})$ and $
	v_{xy} = v_{x} = v_{y}$, then $V(\Gamma_{u}) =^d V(\Gamma_0)$ for any $u = 1, 2, \cdots, (n+m)!$ and so
	$$ 
	\left(V(\Gamma_u)  , Q_{1-\alpha} \left\{  V(\Gamma_1), \cdots, V(\Gamma_{(n+m)!}) \right\}\right) = ^d\left(V(\Gamma_0)  , Q_{1-\alpha} \left\{  V(\Gamma_1), \cdots, V(\Gamma_{(n+m)!}) \right\}\right) .
	$$ 
	Thus, we have
	\begin{multline*}
	P\left(V(\Gamma_0)  > Q_{1-\alpha} \left\{  V(\Gamma_1), \cdots, V(\Gamma_{(n+m)!}) \right\} \right)  
	 = \\ E \left[  \frac{1}{(n+m)!} \sum\limits_{u=1}^{(n+m)!} \mathbb{I}_{ \left\lbrace V(\Gamma_u)  > Q_{1-\alpha} \left\{  V(\Gamma_1), \cdots, V(\Gamma_{(n+m)!}) \right\} \right\rbrace } \right] \leq \alpha.
	\end{multline*}
	(ii) The proof follows similarly from part (i) by observing that for any $s = 1,2, \cdots, S$
	\begin{multline*}
	\left(V(\mathbf \Gamma_s)  , Q_{1-\alpha} \left\{V(\Gamma_0),  V(\mathbf \Gamma_1), \cdots, V(\mathbf \Gamma_{S-1}) \right\}\right) \\ = ^d\left(V(\Gamma_0)  , Q_{1-\alpha} \left\{V(\Gamma_0),  V(\mathbf \Gamma_1), \cdots, V(\mathbf \Gamma_{S-1}) \right\}\right).
	\end{multline*}

	\subsection{Proof of Theorem \ref{lem:per}}
	\noindent (i) Recall that for a fixed permutation matrix $ \varGamma_w \in \mathbb{S}_{w}$ that corresponds to  $ \varPi_{w,ij} $,
	\begin{align*}
	\text{ED}_n^{k}( \varGamma_w \mathbf{Z}) & = \underbrace{\sum_{i=2}^{n+m}\sum_{j=1}^{i-1}  \varPi_{w,ij} \varphi \left(e_{ij} \right)}_{:=\mu_n (\varGamma_w \mathbf Z)} + \underbrace{\sum_{i=2}^{n+m}\sum_{j=1}^{i-1} \varPi_{w,ij} \mathcal{R}_1(Z_i, Z_j)}_{:=R_1( \varGamma_w \mathbf{ Z} )}.
	\end{align*}
	Under the HDMSS setting, part (i) follows from Lemma \ref{lem:event}.
	
	\begin{lemma} \label{lem:event}
		Under Assumption \ref{ass:2}, 
		$
		\sup_{\Gamma} \left| R_{1}(\Gamma \mathbf{ Z}) \right| = o_p(1).
		$
	\end{lemma}
\begin{proof}
Consider the events $B_{\mathbf{X}\mathbf{Y}}, B_{\mathbf{X}}, B_{\mathbf{Y}} $ and their complements $B^c_{\mathbf{X}\mathbf{Y}}, B_{\mathbf{X}}^c, B_{\mathbf{Y}}^c$, where
	\begin{align*}
	& B_{\mathbf{X}\mathbf{Y}} = \left\lbrace \min\limits_{1\leq s \leq n, 1 \leq t\leq m}  \mathcal{L}(X_{s}, Y_{t})  \leq -\frac{1}{2}e_{xy} \text{ or }\max\limits_{1\leq s \leq n,  1 \leq t\leq m}   \mathcal{L}(X_{s}, Y_{t})   \geq  \frac{1}{2} e_{xy} \right\rbrace, \\
	& B_{\mathbf{X}} = \left\lbrace \min\limits_{1\leq s \neq t \leq n}   \mathcal{L}(X_{s}, X_{t})   \leq -\frac{1}{2}e_x \text{ or } \max\limits_{1\leq s \neq t \leq n}   \mathcal{L}(X_{s}, X_{t})  \geq  \frac{1}{2} e_x  \right\rbrace, \\
	& B_{\mathbf{Y}} = \left\lbrace \min\limits_{1\leq s \neq t \leq m}   \mathcal{L}(Y_{s}, Y_{t})   \leq -\frac{1}{2}e_y  \text{ or } \max\limits_{1\leq s \neq t \leq m}  \mathcal{L}(Y_{s}, Y_{t})   \geq  \frac{1}{2}e_y \right\rbrace. 
	\end{align*}
	Then, under assumption \ref{ass:2}, $ \text{as } n\wedge m \wedge p \rightarrow \infty $
	\begin{align*}
	P(B_{\mathbf{X}\mathbf{Y}}) 
	& = P\left( \bigcup\limits_{1\leq s \leq n, 1 \leq t \leq m} \left\lbrace  \mathcal{L}(X_{s}, Y_{t})   \leq -\frac{1}{2}e_{xy}  \text{ or } \mathcal{L}(X_{s}, Y_{t})  \geq \frac{1}{2} e_{xy} \right\rbrace \right) \\
	& \leq \sum\limits_{1\leq s \leq n, 1 \leq t \leq m} P \left( \left| \mathcal{L}(X_{s}, Y_{t}) \right| \geq \frac{1}{2} e_{xy} \right) \\
	& \leq nm P \left(  |\mathcal{L}(X, Y) |  \geq \frac{1}{2} e_{xy} \right) \\
	& \leq \frac{ 4 nm E\left[ \mathcal{L}(X, Y)^2 \right] }{ e_{xy}^2 }  \\
	& = o(1) .
	\end{align*}
	Similarly, we can show that $ P(B_{\mathbf{X}}) = o(1) $ and $  P(B_{\mathbf{Y}}) = o(1) $. Conditioned on event $ B^c_{\mathbf{X}\mathbf{Y}}B^c_{\mathbf{X}}B^c_{\mathbf{Y}} $, we have $  e_{ij} \leq e_{ij} + v \mathcal{L}(Z_i, Z_{j}) \leq 3 e_{ij}/2 $ for any $0 \leq v \leq 1$. Suppose $\varphi^{(1)}(\cdot)$ is a continuous function on $(0, + \infty)$, we know there exist a constant $C$ such that $| \varphi^{(1)} ( e_{ij} + v \mathcal{L}(Z_i, Z_{j}) )| \leq C $ and consequently, we have
	\begin{align*}
	\begin{array}{ll}
	|\mathcal{R}_1(Z_i, Z_j)| \leq C' |\mathcal{L}(Z_i, Z_{j})|,
	\end{array}
	\end{align*}
	where $C'$ is a constant depends only on $\varphi, e_{xy}, e_x$ and $e_y$. Let $\Pi_{ij}$ corresponds to $\Gamma$,
	\begin{align*}
	& \sup_{\Gamma} \left|  \sum_{i=2}^{n+m}\sum_{j=1}^{i-1}  \Pi_{ij} \mathcal{R}_1(Z_i, Z_j) \right| \\
	\leq & \sup_{\Gamma} \left|  \sum_{i=2}^{n+m}\sum_{j=1}^{i-1}  \Pi_{ij} \mathcal{R}_1(Z_i, Z_j) \left\lbrace  \mathbb{I}_{B^c_{\mathbf{X}\mathbf{Y}}} \mathbb{I}_{B^c_{\mathbf{X}}} \mathbb{I}_{B^c_{\mathbf{Y}}} + \mathbb{I}_{B_{\mathbf{X}\mathbf{Y}}} +  \mathbb{I}_{B_{\mathbf{X}}} +  \mathbb{I}_{B_{\mathbf{Y}}} \right\rbrace  \right| \\
	\leq &  \sup_{\Gamma} \sum_{i=2}^{n+m}\sum_{j=1}^{i-1} | \Pi_{ij} \mathcal{R}_1(Z_i, Z_j)| \mathbb{I}_{ B_{\mathbf{X}\mathbf{Y}}^c}\mathbb{I}_{B^c_{\mathbf{X}}}\mathbb{I}_{B^c_{\mathbf{Y}}} + o_p(1) \\
	\leq &   \sum_{i=2}^{n+m}\sum_{j=1}^{i-1} \frac{C''}{nm} | \mathcal{R}_1(Z_i, Z_j)| \mathbb{I}_{ B_{\mathbf{X}\mathbf{Y}}^c}\mathbb{I}_{B^c_{\mathbf{X}}}\mathbb{I}_{B^c_{\mathbf{Y}}} + o_p(1),
	\end{align*}
	where $C''$ is a constant depends only on $\rho$. Then, for any $\epsilon > 0$, by Markov's inequality
	\begin{multline*}
	 P \left(  \frac{1 }{mn}  \sum_{i=2}^{n+m}\sum_{j=1}^{i-1} |\mathcal{R}_1(Z_i, Z_j)|  \mathbb{I}_{ B_{\mathbf{X}\mathbf{Y}}^c}\mathbb{I}_{B^c_{\mathbf{X}}}\mathbb{I}_{B^c_{\mathbf{Y}}} > \epsilon  \right) \\ \leq   \frac{1}{\epsilon mn}   \sum_{i=2}^{n+m}\sum_{j=1}^{i-1}  E[|\mathcal{R}_1(Z_i, Z_j)|  \mathbb{I}_{ B_{\mathbf{X}\mathbf{Y}}^c}\mathbb{I}_{B^c_{\mathbf{X}}}\mathbb{I}_{B^c_{\mathbf{Y}}}] 
	\leq  C''' \frac{  (\alpha_{xy} + \alpha_x + \alpha_y)}{\epsilon},
	\end{multline*}
	where $C'''$ is a constant depends only on $ \rho$, $\varphi, e_{xy}, e_x$ and $e_y$.
\end{proof}
(ii) Similar to the HDLSS setting, we consider the following decomposition
\begin{multline*}
\sqrt{nmp}   \text{ED}_n^{k}( \varGamma_w \mathbf{Z}) =  \sqrt{nmp} \mu_n ( \varGamma_w \mathbf Z)  + \underbrace{\sqrt{nm} \sum_{i=2}^{n+m}\sum_{j=1}^{i-1}   \varPi_{w,ij}\varphi^{(1)}\left( e_{ij} \right) \mathcal{K}(Z_i, Z_j)}_{:=\sqrt{nm}L( \varGamma_w \mathbf{ Z} )}  \\ + \underbrace{\sqrt{nm}\sum_{i=2}^{n+m}\sum_{j=1}^{i-1}\varPi_{w,ij} \varphi^{(1)}(e_{ij}) \mathcal{W}(Z_i, Z_j)}_{:=\sqrt{nm} R_l( \varGamma_w \mathbf{ Z})} + \underbrace{ \sqrt{nmp}\sum_{i=2}^{n+m}\sum_{j=1}^{i-1}  \varPi_{w,ij}\mathcal{R}_2(Z_i,Z_j)}_{:= \sqrt{nm} R_{2}( \varGamma_w \mathbf{ Z}) }.
\end{multline*}
Next, for any $- \infty < a < \infty$ and $\epsilon>0$, using the inequality $P(X \leq a) \leq P(Y \leq a + \epsilon) + P(|X-Y| > \epsilon)$, we can show that 
\begin{multline*}
P(\sqrt{nm}L(\varGamma_w \mathbf Z) \leq a - \epsilon) - P(|\sqrt{nm} R_l(\varGamma_w \mathbf{ Z}) + \sqrt{nm} R_2(\varGamma_w \mathbf{ Z})| > \epsilon) \\ \leq P(\sqrt{nmp} [  \text{ED}_n^{k}(\varGamma_w \mathbf{Z}) - \mu_n (\varGamma_w \mathbf Z) ] \leq a  ) \\ \leq P(\sqrt{nm} L(\varGamma_w \mathbf Z) \leq a + \epsilon) + P(| \sqrt{nm}R_l(\varGamma_w \mathbf{ Z}) + \sqrt{nm}R_2(\varGamma_w \mathbf{ Z})| > \epsilon).
\end{multline*}
Then, some algebra shows that 
\begin{align*}
& \sup_w \left|P(\sqrt{nmp} [  \text{ED}_n^{k}(\varGamma_w \mathbf{Z}) - \mu_n (\varGamma_w \mathbf Z) ] \leq a  ) - \Phi \left(a / \sqrt{nm\sigma^2_n(\varGamma_w \mathbf{Z})} \right)  \right| \\
\leq & \sup_w \left| P(\sqrt{nm} L(\varGamma_w \mathbf Z) \leq a - \epsilon) - \Phi \left( (a-\epsilon) / \sqrt{nm\sigma^2_n(\varGamma_w \mathbf{Z})} \right) \right|  \\
& \hspace{2cm} +  \sup_w \left|  \Phi \left( a / \sqrt{nm\sigma^2_n (\varGamma_w \mathbf{Z})} \right) - \Phi \left( (a-\epsilon) /\sqrt{nm \sigma^2_n (\varGamma_w \mathbf{Z})} \right) \right| \\
+ & \sup_w \left| P(\sqrt{nm} L(\varGamma_w \mathbf Z) \leq a + \epsilon) - \Phi \left( (a+\epsilon) /\sqrt{nm\sigma^2_n(\varGamma_w \mathbf{Z})} \right) \right|  \\
& \hspace{2cm} +  \sup_w \left|  \Phi \left( a /\sqrt{nm \sigma^2_n(\varGamma_w \mathbf{Z})} \right) - \Phi \left( (a+\epsilon) / \sqrt{nm\sigma^2_n(\varGamma_w \mathbf{Z})} \right) \right| \\
+ & 2 \sup_w P(|\sqrt{nm}R_l(\varGamma_w \mathbf{ Z}) +\sqrt{nm} R_2(\varGamma_w \mathbf{ Z})| > \epsilon).
\end{align*}
Next, by Lemma \ref{lem:high2}, \ref{lemma:lowRl} and \ref{lem:high3}.  the right hand side can be made arbitrarily small by first choose $\epsilon$ small enough, then $n,m,p$ large enough.

\begin{lemma} \label{lem:high2}
	Under Assumption \ref{ass:2} and \ref{ass:3}, 
	$
	\sup_{\Gamma} \left|\sqrt{nm} R_{2}(\Gamma \mathbf{ Z}) \right| = o_p(1).
	$
\end{lemma}
\begin{proof}
	The proof is similar with Lemma \ref{lem:event} by observing that conditioned on event $B^c_{\mathbf{X}\mathbf{Y}}B^c_{\mathbf{X}}B^c_{\mathbf{Y}} $, it holds for some constant $C$ that
	$
	|\mathcal{R}_2(Z_i,Z_j)| \leq C \left| \mathcal{L}^2(Z_i, Z_j) \right|.
	$
\end{proof}

	\begin{lemma} \label{lemma:lowRl} Under $H_{A_l}$, 
	$
	\sup_{\Gamma \in \mathbb{P}_{n+m}}|\sqrt{nm}R_l( \Gamma \mathbf{ Z})| = o_p(1).
	$
\end{lemma}

\begin{proof}
	For any fixed permutaiton marix $\Gamma \in \mathbb{P}_{n+m}$, we have
	\begin{multline*}
	\sqrt{nm}R_l( \Gamma \mathbf{ Z}) 
	=  \sqrt{nmp} \sum_{i=2}^{n+m}\sum_{j=1}^{i-1}  \Pi_{ij}\varphi^{(1)}\left(e_{ij}\right) \left(E\left[ \overline{\psi} (Z_{i}, Z_{j})| Z_{i}\right] +  E\left[ \overline{\psi} (Z_{i}, Z_{j})| Z_{j}\right]  \right)  \\
	- \sqrt{nmp} \sum_{i=2}^{n+m}\sum_{j=1}^{i-1}  \Pi_{ij}\varphi^{(1)}\left(e_{ij} \right) \left(E\left[\overline{\psi}(Z_{i}, Z_{j})\right] + e_{ij} \right) .
	\end{multline*}
	Let $w = N( \Gamma)$, similar to the computation of $\mu_{n,w}$, we obtain
	\begin{multline*}
	\sqrt{nmp}\sum_{i=2}^{n+m}\sum_{j=1}^{i-1}  \Pi_{ij}\varphi^{(1)}\left(e_{ij}\right)  E\left[\overline{\psi}(Z_{i}, Z_{j}) \right]  =  \sqrt{nmp}   \bigg\{  2\varphi^{(1)}\left(e_{xy} \right) E \left[ \overline{\psi} (X, Y) \right] \\   
	- \varphi^{(1)}\left(e_{x}\right)  E \left[ \overline{\psi} (X, X')\right]  - \varphi^{(1)}\left( e_{y}\right)  E \left[ \overline{\psi} (Y, Y')\right] \bigg\}  f (w),
	\end{multline*}
	where the right hand side is of order $o_p(1)$ under $H_{A_l}$. Let $\pi$ corresponds to $\Gamma$, then for each $1 \leq i \leq n$ such that $1 \leq  \pi(i) \leq n$, it follows from the definition of ${\Pi}_{ij}$ that
	\begin{align*}
	& \frac{1}{4} \sum_{1 \leq j \leq n+m}^{j \neq i}  \Pi_{ij}\varphi^{(1)}\left(e_{ij}\right) E\left[ \overline{\psi} (X_{i}, Z_{j})| X_{i}\right] \\
	= & -\frac{ (n-w-1)}{n(n-1)} \varphi^{(1)}\left( e_{x}\right) E\left[ \overline{\psi} (X_{i}, X)| X_{i}\right]  - \frac{w}{n(n-1)} \varphi^{(1)}\left(e_{xy}\right) E\left[ \overline{\psi} (X_{i}, Y)| X_{i}\right]  \\
	& \hspace{1cm} +  \frac{w}{nm}  \varphi^{(1)}\left(e_{x}\right) E\left[ \overline{\psi} (X_{i}, X)| X_{i}\right]  + \frac{m- w }{nm}  \varphi^{(1)}\left( e_{xy} \right) E\left[ \overline{\psi} (X_{i}, Y)| X_{i}\right] \\
	= &\left(  \frac{1}{n} - \frac{w}{nm}-\frac{w}{n(n-1)} \right) \bigg\{ \varphi^{(1)}\left(e_{xy}\right) E\left[ \overline{\psi} (X_{i}, Y)| X_{i}\right]  -  \varphi^{(1)}(e_x) E\left[ \overline{\psi} (X_{i}, X)| X_{i}\right] \bigg\},
	\end{align*} 
	which entails 
	\begin{multline*}
	\sup_{\Gamma}\left| \frac{1}{4} \sum_{1 \leq j \leq n+m}^{j \neq i} \Pi_{ij}\varphi^{(1)}\left(e_{ij}\right) E\left[ \overline{\psi} (X_{i}, Z_{j})| X_{i}\right]  \right| \leq \\ \frac{C}{\sqrt{nm}} \left| \varphi^{(1)}\left(e_{xy}\right) E\left[ \overline{\psi} (X_{i}, Y)| X_{i}\right] -  \varphi^{(1)}(e_x) E\left[ \overline{\psi} (X_{i}, X)| X_{i}\right] \right|,
	\end{multline*}
	where $C$ is a constant that only depends on $\rho$. Using the same approach, the above bound can be shown to hold for each $1 \leq i \leq n$ such that $n+1 \leq  \pi(i) \leq n+m$. Similarly, we can show that for each $ n+1 \leq i \leq n+m $,
	\begin{multline*}
	\sup_{\Gamma} \left|  \frac{1}{4} \sum_{1 \leq j \leq n+m}^{j \neq i}  \Pi_{ij}\varphi^{(1)}(e_{ij}) E\left[ \overline{\psi} (Y_{i}, Z_{j})| Y_{i}\right] \right| \\
	\leq   \frac{C}{\sqrt{nm}}  \left| \varphi^{(1)}(e_{xy}) E\left[ \overline{\psi} (Y_{i}, X)| Y_{i}\right] - \varphi^{(1)}(e_{y}) E\left[ \overline{\psi} (Y_{i}, Y)| Y_{i}\right] \right|. 
	\end{multline*}	
	Consequently, the following bound holds
	\begin{multline*}
	\sup_{\Gamma}\left| \sqrt{nmp} \sum_{i=2}^{n+m}\sum_{j=1}^{i-1}  \Pi_{ij}\varphi^{(1)}(e_{ij}) \left(E\left[ \overline{\psi} (Z_{i}, Z_{j})| Z_{i}\right] +  E\left[ \overline{\psi} (Z_{i}, Z_{j})| Z_{j}\right]  \right) \right| \\
	\leq C' \sqrt{p} \sum_{i=1}^{n} \left| \varphi^{(1)}(e_{xy}) E\left[ \overline{\psi} (X_{i}, Y)| X_{i}\right] -  \varphi^{(1)}(e_x) E\left[ \overline{\psi} (X_{i}, X)| X_{i}\right] \right| \\
	+ C' \sqrt{p} \sum_{i=n+1}^{n+m} \left| \varphi^{(1)}(e_{xy}) E\left[ \overline{\psi} (Y_{i}, X)| Y_{i}\right] - \varphi^{(1)}(e_{y}) E\left[ \overline{\psi} (Y_{i}, Y)| Y_{i}\right] \right|,
	\end{multline*} 
	where $C'$ is a constant. Finally, an application of Markov's inequality shows that the right hand side is of order $o_p(1)$ under $H_{A_l}$. 
\end{proof}

	\begin{lemma}\label{lem:high3} Under Assumptions \ref{ass:1}.  Let $ \Gamma_{1}, \Gamma_{2} \in \mathbb{P}_{n+m}$. Then, for any constants $a_1,a_2, b$, we have
	\begin{align*}
	& \sup_{\Gamma_1, \Gamma_2} \left| P(a_1\sqrt{nm} L(\Gamma_{1} \mathbf Z) +a_2\sqrt{nm} L(\Gamma_{2} \mathbf Z) \leq b ) - \Phi \left( \frac{b }{\eta_{a_1, a_2}(\Gamma_1, \Gamma_2) } \right) \right| \\
	& \hspace{1cm} \leq
	C \bigg\{  \max\limits_{\Lambda_1, \Lambda_2\in \{ X,Y \} } E \left[ \mathcal{K}^4(\Lambda_1, \Lambda_2') \right]/n^2 \\ &\hspace{1cm} \;\;\; + \max\limits_{\Lambda_1, \Lambda_2, \Lambda_3\in \{ X,Y \} } E \left[ \mathcal{K}^2(\Lambda_1, \Lambda_3'') \mathcal{K}^2(\Lambda_2', \Lambda_3'') \right]/n \\ &\hspace{1cm} \;\;\; + \max_{\Lambda_1, \Lambda_2, \Lambda_3, \Lambda_4\in \{ X,Y \} } E \left[  \mathcal{K}(\Lambda_1, \Lambda_3'') \mathcal{K}(\Lambda_1, \Lambda_4''')  \mathcal{K}(\Lambda_2', \Lambda_4''') \mathcal{K}(\Lambda_2', \Lambda_3'')  \right] \\ 
	& \hspace{1cm} \;\;\; + (v_x - var [ \mathcal{K}(X, X') ])^2 + (v_y - var[ \mathcal{K}(Y, Y') ])^2 + (v_{xy} - var [ \mathcal{K}(X, Y) ])^2
	\bigg\}^{1/5},
	\end{align*}
	where $C$ is a constant depend on $\varphi, \rho, e_{x},e_{y}$ and $e_{xy}$ only; $\eta_{a_1, a_2}(\Gamma_1, \Gamma_2)$ is defined as
	\begin{align*}
	[\eta_{a_1, a_2}(\Gamma_1, \Gamma_2)]^2 = nm  \sum_{i=2}^{n+m}\sum_{j=1}^{i-1}  (a_1\Pi_{1,ij} + a_2 \Pi_{2,ij} )^2 [\varphi^{(1)}(e_{ij})]^2 v_{ij},
	\end{align*}
	where $\Pi_{1,ij}, \Pi_{2, ij}$ correspond to $\Gamma_1, \Gamma_2$ respectively and
	\begin{align*}
	v_{ij} = \left\lbrace  \begin{array}{ll}
	v_{x}, & \text{ if } 1 \leq i, j \leq n, \\
	v_{y}, & \text{ if } n+1 \leq i, j \leq n+m, \\
	v_{xy}, & \text{ otherwise}.
	\end{array} \right. 
	\end{align*} 
\end{lemma}
\begin{proof}
	Notice that 
	\begin{align*}
	\sqrt{nm} \left( a_1  L(\Gamma_{1} \mathbf Z) +a_2 L(\Gamma_{2} \mathbf Z)\right)  = \sqrt{nm} \sum_{i=2}^{n+m}\sum_{j=1}^{i-1}  (a_1\Pi_{1,ij} + a_2 \Pi_{2,ij} ) \varphi^{(1)}(e_{ij}) \mathcal{K}(Z_i, Z_j).
	\end{align*}
	Then, for notational convenience, set
	$$
	\mathcal{H}(Z_i, Z_j) =  \sqrt{nm} (a_1\Pi_{1,ij} + a_2 \Pi_{2,ij} ) \varphi^{(1)}(e_{ij}) \mathcal{K}(Z_i, Z_j),
	$$ 
	and $S_{n+m,l} = \sum_{i=2}^{l} \xi_{n+m,i}$, where $\xi_{n+m,i} = \sum_{j=1}^{i-1} \mathcal{H}(Z_i, Z_j)$. Next, let 
	$$ 
	\mathcal{F}_{n+m,l} = \sigma(Z_{1}, Z_2 \cdots, Z_{l})
	$$ 
	be the $\sigma$-algebra generated by $ Z_{1}, \cdots, Z_{l} $, we have $\{ S_{n+m, l}, \mathcal{F}_{n+m, l}, 1 \leq l \leq n+m \}$ is a martingale array and thus we can apply the Berry-Esseen type bound for martingale sequences [Theorem 1 of \cite{hall1981rates}]. By setting $m=0$ and $\delta =1$ in Theorem 1 of \cite{hall1981rates}, we compute  
	\begin{align*}
	\sum_{i = 2}^{ n+m} E \left[  \xi_{n+m,i}^2  \right], var \left[ \sum_{i = 2}^{ n+m} E \left[\left.  \xi_{n+m,i}^2 \right| \mathcal{F}_{n+m,i-1} \right]\right] \text{ and } \sum_{i = 2}^{ n+m} E \left[  \xi_{n+m,i}^4  \right] 
	\end{align*}
	Firstly, due to the property of double centering, E$[ \mathcal{H}(Z_i, Z_j)\mathcal{H}(Z_{i'}, Z_{j'}) ] \neq 0$ only when $ \{i,j\} = \{i',j'\} $. Then, let $\eta_{a_1, a_2}(\Gamma_1, \Gamma_2)$ be in Theorem 1 of \cite{hall1981rates}
	\begin{align*}
	\frac{\eta^2_{a_1, a_2}(\Gamma_1, \Gamma_2)}{nm} := & \lim\limits_{p \rightarrow \infty} \frac{\sum_{i = 2}^{ n+m} E \left[  \xi_{n+m,i}^2  \right]}{nm} =  \sum_{i=2}^{n+m}\sum_{j=1}^{i-1}  (a_1\Pi_{1,ij} + a_2 \Pi_{2,ij} )^2 [\varphi^{(1)}(e_{ij})]^2 v_{ij}.
	\end{align*}
	Then, to calculate the variance, notice that 
	\begin{align*}
	var \left[ \sum_{i = 2}^{ n+m} E \left[\left.  \xi_{n+m,i}^2 \right| \mathcal{F}_{n+m,i-1} \right]\right] = \sum_{i_1, i_2=2}^{n+m} \sum_{j_1, j_2=1}^{i_1-1}  \sum_{j_3, j_4=1}^{i_2-1} \Theta(i_1, i_2; j_1, j_2, j_3, j_4),
	\end{align*}
	where $\Theta(i_1, i_2; j_1, j_2, j_3, j_4)$ is defined as
	\begin{multline*}
	\Theta(i_1, i_2; j_1, j_2, j_3, j_4) = \\ cov\left[E\left[ \mathcal{H}(Z_{i_1}, Z_{j_1}) \mathcal{H}(Z_{i_1}, Z_{j_2})| Z_{j_1}, Z_{j_2}\right], E\left[ \mathcal{H}(Z_{i_2}, Z_{j_3}) \mathcal{H}(Z_{i_2}, Z_{j_4})| Z_{j_3}, Z_{j_4}\right] \right].
	\end{multline*}
	Next, for any $ 1 \leq j_1, j_2 \leq n+m $, $\Lambda \in \{ X,Y \}$, denote 
	\begin{align*}
	\mathcal{G}_{\Lambda}(Z_{j_1}, Z_{j_2}) &= E[  \mathcal{K}(\Lambda, Z_{j_1}) \mathcal{K}(\Lambda, Z_{j_2}) | Z_{j_1}, Z_{j_2} ].
	\end{align*} 
	To bound each $ \Theta(i_1, i_2; j_1, j_2, j_3, j_4)$, we need to study the covariance between $ \mathcal{G}_{\Lambda_1}(Z_{j_1}, Z_{j_2})$ and $\mathcal{G}_{\Lambda_2}(Z_{j_1'}, Z_{j_2'}) $.
	\begin{lemma} \label{lemma:cov}
		Then, for any $ 1 \leq j_1, j_2, j_1', j_2' \leq n+m $, $\Lambda_1, \Lambda_2 \in \{ X,Y \}$,  we have
		\begin{multline*}
		cov\left[ \mathcal{G}_{\Lambda_1}(Z_{j_1}, Z_{j_2}), \mathcal{G}_{\Lambda_2}(Z_{j_1'}, Z_{j_2'}) \right] \\ =  \left\lbrace 
		\arraycolsep=1.4pt\def\arraystretch{1.6}
		\begin{array}{ll}
		E \left[ \mathcal{K}^2(\Lambda_1, Z_{j_1}) \mathcal{K}^2(\Lambda_2', Z_{j_1}) \right] - E \left[ \mathcal{K}^2(\Lambda_1,Z_{j_1}) \right] E \left[ \mathcal{K}^2(\Lambda_2, Z_{j_1}) \right] , & j_1=j_2=j_1'=j_2'; \\
		E \left[  \mathcal{K}(\Lambda_1, Z_{j_1}) \mathcal{K}(\Lambda_1, Z_{j_2})  \mathcal{K}(\Lambda_2', Z_{j_2}) \mathcal{K}(\Lambda_2', Z_{j_1})  \right] , & j_1=j_1'\neq j_2=j_2'; \\
		E \left[  \mathcal{K}(\Lambda_1, Z_{j_1}) \mathcal{K}(\Lambda_1, Z_{j_2})  \mathcal{K}(\Lambda_2', Z_{j_2}) \mathcal{K}(\Lambda_2', Z_{j_1})  \right] , & j_1=j_2'\neq j_2=j_1'; \\
		0, & \text{otherwise}.
		\end{array}
		\right. 
		\end{multline*}
	\end{lemma}
	\begin{proof}
		If $j_1 = j_2' \neq j_2 = j_1'$, 
		\begin{align*}
		& E \left[\mathcal{G}_{\Lambda_1}(Z_{j_1}, Z_{j_2})  \mathcal{G}_{\Lambda_2}(Z_{j_1'}, Z_{j_2'}) \right] \\
		= & E \left[ E\left[ \left. \mathcal{K}(\Lambda_1, Z_{j_1}) \mathcal{K}(\Lambda_1, Z_{j_2}) \right| Z_{j_1}, Z_{j_2}\right] E\left[ \left.  \mathcal{K}(\Lambda_2', Z_{j_1'}) \mathcal{K}(\Lambda_2', Z_{j_2'}) \right| Z_{j_1}, Z_{j_2} \right]\right] \\
		= &  E \left[ E\left[ \left.  \mathcal{K}(\Lambda_1, Z_{j_1}) \mathcal{K}(\Lambda_1, Z_{j_2})  \mathcal{K}(\Lambda_2', Z_{j_2}) \mathcal{K}(\Lambda_2', Z_{j_1}) \right|  Z_{j_1}, Z_{j_2} \right] \right]\\
		= & E \left[  \mathcal{K}(\Lambda_1, Z_{j_1}) \mathcal{K}(\Lambda_1, Z_{j_2})  \mathcal{K}(\Lambda_2', Z_{j_2}) \mathcal{K}(\Lambda_2', Z_{j_1})  \right] 
		\end{align*}
		It can be shown similarly for cases $ j_1=j_2=j_3=j_4$ and $  j_1=j_1'\neq j_2=j_2'$. Next, we show that for other cases, the covariance is 0. We take $j_1=j_1'$, $j_1 \neq j_2$, $j_1 \neq j_2'$, $j_2 \neq j_2'$ as an example 
		\begin{align*}
		& E \left[\mathcal{G}_{\Lambda_1}(Z_{j_1}, Z_{j_2})  \mathcal{G}_{\Lambda_2}(Z_{j_1'}, Z_{j_2'}) \right] \\
		= & E \left[ E\left[ \left. \mathcal{K}(\Lambda_1, Z_{j_1}) \mathcal{K}(\Lambda_1, Z_{j_2}) \right| Z_{j_1}, Z_{j_2}\right] E\left[ \left. \mathcal{K}(\Lambda_2', Z_{j_1}) \mathcal{K}(\Lambda_2', Z_{j_2'}) \right| Z_{j_1}, Z_{j_2'} \right] \right] \\
		= &  E \left[ E\left[ \left.  \mathcal{K}(\Lambda_1, Z_{j_1}) \mathcal{K}(\Lambda_1, Z_{j_2})  \mathcal{K}(\Lambda_2', Z_{j_1}) \mathcal{K}(\Lambda_2', Z_{j_2'}) \right|  Z_{j_1}, Z_{j_2}, Z_{j_2'} \right] \right] \\
		= &  E \left[  \mathcal{K}(\Lambda_1, Z_{j_1}) \mathcal{K}(\Lambda_1, Z_{j_2})  \mathcal{K}(\Lambda_2', Z_{j_1}) \mathcal{K}(\Lambda_2', Z_{j_2'})  \right] \\
		= & E \left[\mathcal{K}(\Lambda_1, Z_{j_1}) \mathcal{K}(\Lambda_2', Z_{j_1})E\left[ \left.  \mathcal{K}(\Lambda_1, Z_{j_2})    \right| \Lambda_1, \Lambda_2', Z_{j_1} \right]E \left[ \left. \mathcal{K}(\Lambda_2', Z_{j_2'}) \right| \Lambda_1, \Lambda_2', Z_{j_1} \right]  \right] \\
		= & 0.
		\end{align*}
	\end{proof}
	Next, we can bound $var \big[ \sum_{i = 2}^{ n+m} E \left[\left.  \xi_{n+m,i}^2 \right| \mathcal{F}_{n+m,i-1} \right]\big]$ as
	\begin{align*}
	&\sum_{i_1, i_2=1}^{n+m} \sum_{j_1, j_2=1}^{i_1-1}  \sum_{j_3, j_4=1}^{i_2-1} \Theta(i_1, i_2; j_1, j_2, j_3, j_4) \\
	= & \sum_{i = 1}^{n+m} \Bigg\{  \sum_{j=1}^{i-1} \Theta(i, i; j, j, j, j)  + 2 \sum_{1 \leq j_1 \neq j_2 \leq i-1} \Theta(i, i; j_1, j_2, j_1, j_2) \Bigg\} \\
	&\hspace{1cm} +2 \sum_{1 \leq i_1 <i_2 \leq n+m } \Bigg\{  \sum_{j=1}^{i_1-1} \Theta(i_1, i_2; j, j, j, j)  + 2 \sum_{1 \leq j_1 \neq j_2 \leq i_1-1} \Theta(i_1, i_2; j_1, j_2, j_1, j_2) \Bigg\} \\
	= & O\Big( \max_{\Lambda_1, \Lambda_2, \Lambda_3\in \{ X,Y \} } E \left[ \mathcal{K}^2(\Lambda_1, \Lambda_3'') \mathcal{K}^2(\Lambda_2', \Lambda_3'') \right]/n \\
	& \hspace{1cm} + \max_{\Lambda_1, \Lambda_2, \Lambda_3, \Lambda_4\in \{ X,Y \} } E \left[  \mathcal{K}(\Lambda_1, \Lambda_3'') \mathcal{K}(\Lambda_1, \Lambda_4''')  \mathcal{K}(\Lambda_2', \Lambda_4''') \mathcal{K}(\Lambda_2', \Lambda_3'')  \right]\Big) 
	\end{align*}
	Finally, to find the upper bound of $\sum_{i = 2}^{ n+m} E \left(  \xi_{n+m,i}^4  \right)$,
	\begin{align*}
	& \sum_{i = 2}^{ n+m} E \left(  \xi_{n+m,i}^4  \right) \\
	= & \sum_{i = 2}^{ n+m} \sum_{j_1, j_2, j_3, j_4=1}^{i-1} E\left[ \mathcal{H}(Z_{i}, Z_{j_1}) \mathcal{H}(Z_{i}, Z_{j_2}) \mathcal{H}(Z_{i}, Z_{j_3}) \mathcal{H}(Z_{i}, Z_{j_4}) \right] \\
	= & \sum_{i = 2}^{ n+m} \sum_{j=1}^{i-1} E\left[ \mathcal{H}^4(Z_{i}, Z_{j})  \right] + 6 \sum_{i = 2}^{ n+m} \sum_{1 \leq j_1 < j_2 \leq i-1}  E\left[ \mathcal{H}^2(Z_{i}, Z_{j_1}) \mathcal{H}^2(Z_{i}, Z_{j_2}) \right] \\
	= &  O\left(  \max_{\Lambda_1, \Lambda_2\in \{ X,Y \} } E \left[ \mathcal{K}^4(\Lambda_1, \Lambda_2') \right]/n^2 \right) \\
	&  \hspace{1cm}+ O\left(\max_{\Lambda_1,\Lambda_2, \Lambda_3\in \{ X,Y \} } E \left[ \mathcal{K}^2(\Lambda_1, \Lambda_3'') \mathcal{K}^2(\Lambda_2', \Lambda_3'') \right]/n \right).
	\end{align*}  
	Combining the above bounds, the lemma is a consequence of Theorem 1 in \cite{hall1981rates}.
\end{proof}

\subsection{Proof of Theorem \ref{lem:rand}}
(i) For a random permutation matrix $\mathbf{ \Gamma} \sim \text{Uniform}(\mathbb P_{n+m})$, it follows from Lemma \ref{lem:event} that
\begin{align*}
\text{ED}_n^{k}( \mathbf \Gamma \mathbf{Z})  = \mu_n (\mathbf \Gamma \mathbf Z) + R_{1}(\mathbf \Gamma \mathbf Z )   = \left( 2\varphi \left(e_{xy} \right) - \varphi \left(e_x \right) - \varphi \left(e_y \right) \right) f(W) + o_p(1)
\end{align*}
where $W = N(\mathbf{\Gamma}) \sim \text{Hypergeometric}(m+n,m,n)$.  From the normal limit of hypergeometric distribution \cite{lahiri2006sub}, we know that
	\begin{align*}
	\frac{W}{\sqrt{nm}} \overset{p}{\rightarrow} \frac{\sqrt{\rho}}{1 + \rho}.
	\end{align*}
	Next, some algebra shows that $ f(W) \overset{p}{\rightarrow} 0$ and so the result is proved. 
	
	\noindent (ii) 
	Recall that we can decompose the sample energy distance as 
	\begin{multline*}
	\sqrt{nmp}   \text{ED}_n^{k}(\mathbf \Gamma \mathbf{Z}) =  \sqrt{nmp} \mu_n (\mathbf \Gamma \mathbf Z)  + \underbrace{\sqrt{nm} \sum_{i=2}^{n+m}\sum_{j=1}^{i-1}   \mathbf \Pi_{ij}\varphi^{(1)}\left( e_{ij} \right) \mathcal{K}(Z_i, Z_j)}_{:= \sqrt{nm}L(\mathbf{\Gamma} \mathbf{ Z} )}  \\ + \underbrace{\sqrt{nm}\sum_{i=2}^{n+m}\sum_{j=1}^{i-1} \mathbf \Pi_{ij} \varphi^{(1)}(e_{ij}) \mathcal{W}(Z_i, Z_j)}_{:= \sqrt{nm}R_l(\mathbf \Gamma \mathbf{ Z})} + \underbrace{ \sqrt{nmp}\sum_{i=2}^{n+m}\sum_{j=1}^{i-1} \mathbf \Pi_{ij}\mathcal{R}_2(Z_i,Z_j)}_{:=\sqrt{nm} R_{2}( \mathbf \Gamma \mathbf{ Z}) }.
	\end{multline*}
	The result is a consequence of Lemma \ref{lemma:lowRl}, \ref{lem:high2} and \ref{lem:joint}.
	
		\begin{lemma} \label{lem:joint}
		Under Assumptions \ref{ass:1} and \ref{ass:4},
		\begin{align*}
		\sqrt{nm} \left( \begin{array}{l}
		L(\mathbf \Gamma \mathbf{Z}) \\
		L(\mathbf \Gamma' \mathbf{Z}) 
		\end{array}  \right) \overset{d}{\rightarrow} N \left(0,  
		\left( \begin{array}{cc}
		\sigma^2 &  0 \\
		0 & \sigma^2
		\end{array} \right) \right) 
		\end{align*}
		where $\mathbf \Gamma'$ is independent copy of $\mathbf \Gamma$ and $\sigma^2$ is the asymptotic variance defined as
		\begin{align*}
		\sigma^2 := 4 v_{xy} [\varphi^{(1)}(e_{xy})]^2 + 2 \rho v_x[ \varphi^{(1)}(e_{x})]^2 + \frac{2}{\rho} v_y[ \varphi^{(1)}(e_{y})]^2.
		\end{align*}	
	\end{lemma}
	\begin{proof} We apply the Cram\'{e}r-Wold device. For any constants $a_1, a_2$, we have
		\begin{align*}
		\eta^2_{a_1, a_2}(\mathbf{\Gamma},\mathbf{\Gamma}' )  & =  nm  \sum_{i=2}^{n+m}\sum_{j=1}^{i-1}  (a_1 \mathbf \Pi_{ij} + a_2 \mathbf \Pi_{ij}' )^2 [\varphi^{(1)}(e_{ij})]^2 v_{ij} \\
		& = nm  \sum_{i=2}^{n+m}\sum_{j=1}^{i-1} \left(a_1^2 \mathbf \Pi_{ij} ^2 + a_2^2 \mathbf (\mathbf \Pi_{ij}')^2 + 2a_1 a_2 \mathbf \Pi_{ij} \mathbf \Pi_{ij}' \right)   [\varphi^{(1)}(e_{ij})]^2 v_{ij} .
		\end{align*}
		Notice that for $\{ i_1, j_1\} \cap \{ i_2, j_2 \} = \emptyset$, it can be shown that
		$
		E[ \mathbf \Pi_{i_1j_1} \mathbf \Pi_{i_2j_2}] = O(1/n^5). 
		$
		Then, denote $c_{ij} = 2a_1a_2 [\varphi^{(1)}(e_{ij})]^2 v_{ij}$, we have
		\begin{align*}
		E \left[ \left( nm \sum_{i=2}^{n+m}\sum_{j=1}^{i-1} c_{ij} \mathbf \Pi_{ij} \mathbf \Pi_{ij}' \right)^2 \right] 
		= & n^2m^2 \sum_{i=2}^{n+m}\sum_{j_1=1}^{i-1} \sum_{j_2=1}^{i-1}  c_{ij_1}c_{ij_2}  E^2[\mathbf \Pi_{ij_{1}} \mathbf \Pi_{ij_{2}} ] \\
		& + 2 n^2m^2  \sum_{2 \leq i_1 < i_2 \leq n+m}\sum_{j=1}^{i_1-1}   c_{i_1j}c_{i_2j} E^2[\mathbf \Pi_{i_1j} \mathbf \Pi_{i_2j} ] \\
		& + n^2m^2 \sum_{2 \leq i_1 \neq i_2 \leq n+m} \sum_{j_1\neq j_2}    c_{i_1j_1}c_{i_2j_2} E^2[\mathbf \Pi_{i_1j_1} \mathbf \Pi_{i_2j_2} ] \\
		= & O(1/n).
		\end{align*}
		In addition, let $W = N(\mathbf \Gamma)$, we obtain
		\begin{align*}
		\sigma^2_n(\mathbf \Gamma \mathbf{Z}) :=& \sum_{i=2}^{n+m}\sum_{j=1}^{i-1}  \mathbf \Pi_{ij} ^2 [\varphi^{(1)}(e_{ij})]^2 v_{ij} \\
		= &  \bigg\{ \frac{4}{nm} -4 \left( \frac{n+m}{n^2m^2} - \frac{n}{n^2(n-1)^2} - \frac{m}{m^2(m-1)^2} \right)W \\ 
		& \hspace{1cm}+ 4 \left(\frac{2}{n^2m^2} - \frac{1}{n^2(n-1)^2} - \frac{1}{m^2(m-1)^2} \right) W^{2} \bigg\} v_{xy} [\varphi^{(1)}(e_{xy})]^2  \\
		+ & \bigg\{ \frac{2}{n(n-1)} + 2\left( \frac{2n}{n^2m^2} - \frac{2n-1}{n^2(n-1)^2} - \frac{1}{m^2(m-1)^2} \right) W \\
		& \hspace{1cm} - 2\left(\frac{2}{m^2n^2} - \frac{1}{n^2(n-1)^2} - \frac{1}{m^2(m-1)^2} \right) W^{2}  \bigg\} v_x[ \varphi^{(1)}(e_{x})]^2 \\
		+ & \bigg\{ \frac{2}{m(m-1)} + 2\left( \frac{2m}{n^2m^2}  - \frac{1}{n^2(n-1)^2} - \frac{2m-1}{m^2(m-1)^2} \right) W \\
		& \hspace{1cm} -2 \left(\frac{2}{n^2m^2} - \frac{1}{m^2(m-1)^2} - \frac{1}{n^2(n-1)^2} \right) W^{2}  \bigg\} v_y[ \varphi^{(1)}(e_{y})]^2.
		\end{align*}
		Since $W/\sqrt{nm} \overset{p}{\rightarrow} \sqrt{\rho}/(1 + \rho)$, some algebra shows that 
		\begin{align*}
		\sigma^2_n(\mathbf \Gamma \mathbf{Z}):=\sum_{i=2}^{n+m}\sum_{j=1}^{i-1}  \mathbf \Pi_{ij} ^2 [\varphi^{(1)}(e_{ij})]^2 v_{ij} \overset{p}{\rightarrow} \sigma^2,
		\end{align*}
		which entails that $\eta^2_{a_1, a_2}(\mathbf \Gamma,\mathbf{\Gamma}') \overset{p}{\rightarrow} a_1^2 \sigma^2 +  a_2^2 \sigma^2$. Since $|\Phi(\cdot)| \leq 1$, we have 
		\begin{align*}
		E \left[\left|\Phi \left(\frac{b}{\eta_{a_1,a_2}(\mathbf \Gamma,\mathbf{\Gamma}')} \right) -  \Phi \left(\frac{b}{\sqrt{a_1^2 \sigma^2 +  a_2^2 \sigma^2}} \right) \right|\right] \rightarrow 0.
		\end{align*}
		Next, by a simple triangle inequality
		\begin{multline*}
		\left| P\left(a_1 \sqrt{nm} L(\mathbf\Gamma \mathbf Z) +a_2 \sqrt{nm} L(\mathbf \Gamma' \mathbf Z) \leq b \right) - \Phi \left(\frac{b}{\sqrt{a_1^2 \sigma^2 +  a_2^2 \sigma^2}} \right) \right|\leq \\ \left| P\left(a_1 \sqrt{nm} L(\mathbf \Gamma \mathbf Z) +a_2 \sqrt{nm} L(\mathbf \Gamma' \mathbf Z) \leq b \right) - \Phi \left(\frac{b}{\sqrt{\eta^2_{a_1,a_2}(\mathbf{\Gamma}, \mathbf{\Gamma}')}} \right) \right| \\ + \left| \Phi \left(\frac{b}{\sqrt{\eta^2_{a_1,a_2}(\mathbf{\Gamma}, \mathbf{\Gamma}')}}\right)- \Phi \left(\frac{b}{\sqrt{a_1^2 \sigma^2 +  a_2^2 \sigma^2}}  \right) \right| .
		\end{multline*}
		Taking expectation with respect to $\mathbf \Gamma, \mathbf{\Gamma}' $ on both sides, then it follows from Lemma \ref{lem:high3} and Assumption \ref{ass:4} that
		\begin{multline*}
	\left| P\left(a_1 \sqrt{nm} L(\mathbf\Gamma \mathbf Z) +a_2 \sqrt{nm} L(\mathbf \Gamma' \mathbf Z) \leq b \right) - \Phi \left(\frac{b}{\sqrt{a_1^2 \sigma^2 +  a_2^2 \sigma^2}} \right) \right|\leq \\ E\left[  \left| P\left(a_1 \sqrt{nm} L(\mathbf \Gamma \mathbf Z) +a_2 \sqrt{nm} L(\mathbf \Gamma' \mathbf Z) \leq b \right) - \Phi \left(\frac{b}{\sqrt{\eta^2_{a_1,a_2}(\mathbf{\Gamma}, \mathbf{\Gamma}')}} \right) \right|\right] \\ + E \left[ \left| \Phi \left(\frac{b}{\sqrt{\eta^2_{a_1,a_2}(\mathbf{\Gamma}, \mathbf{\Gamma}')}}\right)- \Phi \left(\frac{b}{\sqrt{a_1^2 \sigma^2 +  a_2^2 \sigma^2}}  \right) \right|  \right] = o(1).
	\end{multline*}
	\end{proof}

	\subsection{Proof of Corollary \ref{cor:main}}
	By using Theorem 15.2.3 of \cite{lehmann2006testing}, the result is a consequence of Theorem \ref{lem:rand}.
	
	\subsection{Proof of Theorem \ref{lemma:power2}}
	(i) By Corollary \ref{cor:main} and Theorem \ref{lem:per}
	\begin{align*}
	\text{Power}  = P_{H_{A_c}} \left(\text{ED}_n^k(\mathbf{Z}) > c \right)  \rightarrow   P \left( 2\varphi( e_{xy} ) - \varphi(e_{x}) - \varphi (e_y) > 0 \right) =1.
	\end{align*}
	(ii) By Corollary \ref{cor:main} and Theorem \ref{lem:per}
	\begin{multline*}
	\text{Power}  = P_{H_{A_l}} \left(\text{ED}_n^k(\mathbf{Z}) > c \right)  = P_{H_{A_l}} \left(\sqrt{nmp}\text{ED}_n^k(\mathbf{Z})> \sqrt{nmp}  c \right)\\  \rightarrow   P \left( N( 0, \sigma^2) > \sigma Q_{\Phi, 1-\alpha} \right)   = \alpha.
	\end{multline*}

\section*{Acknowledgements}	
We would like to thank Dr. Jun Li for providing the code used in \cite{li2018asymptotic}. We are also grateful to the three reviewers for their very helpful comments. The partial support from a US NSF grant is gratefully acknowledged. 

{
	\bibliographystyle{agsm}
	\bibliography{ro}
}
\end{document}